\documentclass[conference]{IEEEtran}
\usepackage{cite}
\usepackage{url}
\usepackage{amsfonts}
\usepackage{courier}
\usepackage{subfigure}
\usepackage{amsthm}
\usepackage{amssymb}
\usepackage{amsmath}
\usepackage{dsfont}
\usepackage{graphicx} 
\usepackage{mathtools}
\usepackage{epstopdf}
\usepackage{caption}
\usepackage{cases}
\graphicspath{{charts/}} 

\def\baselinestretch{0.98}

\newcommand{\E}{\mathbb{E}}

\newcommand{\RTT}{\mbox{T}}
\newcommand{\VR}{\mbox{VRTT}}

\newtheorem{proposition}{Proposition}

\newtheorem{theorem}{Theorem}
\newtheorem{corollary}{Corollary}[theorem]

\begin{document}

\title{LAC: Introducing Latency-Aware Caching in Information-Centric Networks}

\author{
\IEEEauthorblockN{
Giovanna Carofiglio\IEEEauthorrefmark{1},
Leonce Mekinda\IEEEauthorrefmark{2},
Luca Muscariello\IEEEauthorrefmark{2}
}
\IEEEauthorblockA{
\IEEEauthorrefmark{1} Cisco Systems,
\IEEEauthorrefmark{2} Orange Labs Networks,
\\
gcarofig@cisco.com, 
firstname.lastname@orange.com}
}

\maketitle

\begin{abstract} 
Latency-minimization is recognized as one of the pillars of
5G network architecture design. Information-Centric 
Networking (ICN) appears a promising candidate technology for
building an agile communication model that reduces latency
through in-network caching. However, no proposal has
developed so far latency-aware cache management mechanisms 
for ICN.
In the paper, we investigate the role of latency awareness on
data delivery performance in ICN and introduce LAC, a new
simple, yet very effective, Latency-Aware Cache management
policy. The designed mechanism leverages in a distributed
fashion local latency observations to decide whether to store
an object in a network cache. The farther the object,
latency-wise, the more favorable the caching decision.
By means of simulations, show that LAC outperforms state of 
the art proposals and  results in a reduction of the content 
mean delivery time and standard deviation by up to 50\%, 
along with a very fast convergence to these figures.
\end{abstract}

\section{Introduction}
Latency minimization, or building for virtual zero latency as 
commonly referred to, is one of the pillars of 5G network 
architecture design and is currently fostering important 
research work in this space.
Inserting cache memories across the communication data path 
between different processing elements has been already 
demonstrated to be a reliable way of improving performance 
by localizing - especially popular - content at network edge 
and so reducing retrieval latency.

Besides other advantageous architectural choices, the 
introduction of in-network caching as a native building block 
of the network design makes Information Centric Networking 
(ICN)\cite{Jacobson:2009:NNC:1658939.1658941}, a promising 5G 
network technology.
In a nutshell, every ICN router potentially manages a cache 
of previously requested objects in order to improve object 
delivery by reducing retrieval path length for frequently 
requested content. In fact, if content is locally available 
in the cache, the router sends it back directly to the 
requester, otherwise it forwards the request (or Interest) 
for the object to the next hop according to name-based 
routing criteria. When the requested object comes back, it is 
stored in the local cache before sending it back to the 
requester. Given cache size limitations, a replacement policy 
is put in place to evict previously stored objects for 
accommodating the newly available ones. To this aim, various 
classical cache replacement policies, not specifically 
ICN-based exist: to cite a few, Least-Recently-Used (LRU), 
Least-Frequently-Used (LFU), First-In-First-Out (FIFO) and 
Random (RND) \cite{DBLP:journals/corr/abs-1202-4880}.
Within the panoply of cache management policies proposed in 
the literature, very few exploit object retrieval latency to 
orchestrate cache decisions, while requiring transport 
protocol modifications \cite{Prob-Cache} or involving 
additional computational 
complexity\cite{Tong:2013:LSS:2505515.2507857}, without 
significant caching performance increase.

Clearly, the constraints imposed by ICN in terms of high 
speed packet processing exclude every complex cache 
management policy. Therefore, we focus in this paper on a 
simple, hence feasible, cache management policy leveraging
not only the objects replacement, but the cache insertion 
criterion,  that we define based on monitored object latency.

The cache management mechanism we propose in this paper, LAC, 
lies upon the following principle: every time an object is 
received from the network, it is stored into the cache with a 
probability proportional to its recently observed retrieval 
latency. As such, it is an add-on laying on top of any cache 
replacement policy and feeding it at a regulated pace. 
In this way, LAC implicitly prioritizes long-to-retrieve 
objects, instead of caching every object regardless. 
The underlying trade off such caching mechanism tackles is 
between a limited cache size and delivery time minimization.
As caching intrinsically aims to relieve the fallouts of 
network distance or traffic congestion, it must be aware of 
both delay factors to efficiently handle the cache size / 
delivery time tradeoff.  Data retrieval latency is a simple, 
locally measurable and consistent metric for revealing either 
haul distance or traffic congestion.
More precisely, the contribution of this paper is threefold :

\begin{itemize}
\item We design LAC, a randomized dynamic cache management 
policy leveraging in-network retrieval latency for cache 
insertion. The locally monitored metric is the time elapsing 
at a given node between request forwarding and corresponding 
packet reception.

\item We provide a preliminary analysis of LAC
to prove its superior performance over a symmetric p-LRU 
(probabilistic LRU) policy using 
the same probability $p$ for Move-To-Front (MTF) operation in 
case of hit and miss events. 

\item We evaluate LAC performance
by means of packet level simulations carried out with our ICN simulator CCNPL-Sim (\url{http://systemx.enst.fr/ccnpl-sim}). 
\end{itemize}

The rest of the paper is structured in the following way. We 
review the state of the art and perceived limitations in 
\S\ref{sec:related_work}. The problem formulation of 
latency-aware caching is reported in 
\S\ref{sec:problem_formulation_design_choices}.
Sec. \S\ref{sec:system_model} gathers analytical results, 
while performance evaluation of our proposal is in 
\S\ref{sec:performance_evaluation}.
Finally, Sec.\ref{sec:conclusion} concludes the paper, by 
giving a glimpse on future activities.

\section{Related Work} \label{sec:related_work}
In the context of ICN research, previous work have considered 
the enhancement of cache mechanisms with the aim of reducing 
caching redundancy over a delivery path. We can distinguish 
two categories of related work: those leveraging content 
placement (e.g. \cite{Congestion-Aware-Gerla}, 
\cite{Time-Shifted-Simon}) as opposed to
those proposing caching mechanisms based on selective 
insertion/replacement in cache (e.g. \cite{Prob-Cache}, 
\cite{ICN2014-Kurose}, \cite{LCN2014}, \cite{Age-Based}).
The first class of approaches has a limited applicability to 
controlled environments like a CDN (Content Delivery 
Network), where topology and content catalog are know a 
priori. Either \cite{Congestion-Aware-Gerla} and 
\cite{Time-Shifted-Simon} deals with video streaming 
in ICN and orchestrate caching and scheduling of requests to 
caches in order to create a cluster of caches with a certain 
number of guaranteed replicas (\cite{Time-Shifted-Simon}).
Unlike these approaches, our work belongs to the second class 
of caching solutions and aims at defining a decentralized 
caching solution that automatically adapts to changes in 
content popularity, network variations etc. by leveraging 
content insertion/replacement operations in cache. 
We share the same objective as \cite{ICN2014-Kurose}, where 
authors propose a congestion-aware caching mechanism
for ICN based on estimation of local congestion,  
of popularity and of position w.r.t. the bottleneck. 
The congestion estimate in this work does not allow to 
differentiate content items in terms of latency like in our 
work. A similar consideration holds for other related 
approaches: the ProbCache work in \cite{Prob-Cache}, which 
utilizes the same cache probability for every content item at 
a given node and the cooperative caching mechanism in 
\cite{LCN2014}-\cite{Age-Based} exploiting overall popularity 
and distance-to-server. Clearly, the rationale behind is the 
same, but the distance-to-server metric does not reflects the 
differences in terms of latency, distance to bottleneck on a 
per-flow basis that our approach takes into account. 
Beyond ICN, caching literature is vast 
\cite{Podlipnig:2003:SWC:954339.954341} and our review 
here does not attempt to be exhaustive, while rather to 
position our contribution w.r.t. closest classical caching 
approaches. Starobinski 
\textit{et al.}
\cite{Starobinski:2001:PMW:570289.570293} and later
Jelenkovi\'{c} 
\textit{et al.}\cite{Jelenkovic:2004:OLC:1024662.1024670} 
describe a cache management mechanisms to optimize the 
storage of variable size documents. In their work, the whole 
Move-To-Front rule is \textit{symmetric}, i.e. applied in 
both hit and miss events (as for LRU,LFU etc.)
while our approach, instead, may be denoted as 
\textit{asymmetric}, since it restricts the stochastic 
decision of MTF to cache miss events, leaving object 
replacement subject to deterministic LRU. 

Furthermore, Starobinski \textit{et al.} model only focuses 
on a single cache, while assuming that every object is 
associated with a fixed retrieval cost, just 
like an intrinsic property. However, in a network of caches, 
an object's retrieval cost 
may vary considerably, depending on its current location and 
on network congestion.

\section{Problem formulation and design choices}
\label{sec:problem_formulation_design_choices}
The problem of improving end-user delivery performance can be 
formulated as the minimization of the overall average 
delivery time $\E[\RTT]$ for all users in the network and 
over all requested objects.

\begin{numcases}{}
\min \sum_{u\in \mathcal{U}} \sum_{k\in\mathcal{K} } 
\sum_{r \in \mathcal{R}_{u,k}} q_{k,u} p_{k,r,u} \E[\RTT_{k,r,u}]  \nonumber \\
\sum_{k} q_{k,u} =1, \qquad \forall u  \\
\sum_{r} p_{k,r,u} =1, \qquad \forall k, r   \label{eq:const2}\\
0 \leq q_{k,u} \leq 1 \qquad \forall k, u \\
0 \leq p_{k,r,u} \leq 1 \qquad \forall k, r, u
\end{numcases}
where $q_{k,u}$ is the normalized request rate of object $k$ from user $u$ (namely, the popularity function at user $u$),
and $p_{k,r,u}$ is the probability to download object $k$ 
from route $r$ and $\E[\RTT_{k,r,u}]$ is the average latency 
to retrieve object $k$ on route $r$. The set of routes 
available at user $u$ is identified by $\mathcal{R}_{u,k}$.

Using the Lagrangian of the problem and imposing the 
Karush-Kuhn-Tucker (KKT) optimality conditions it is easy to 
show that for some constants $c_1$, $c_2$.
\begin{equation}\label{eq:objective}
p_{k,r,u} = \frac{c_1}{q_{k,u} \E[\RTT_{k,r,u}]+c_2}
\end{equation} 
In this paper we look for a distributed algorithm that tries 
to minimize this objective by obtaining $p_{k,r,u}$ without 
any coordination among the nodes and no signaling.
The optimal objective expressed in (\ref{eq:const2}]) can be 
heuristically generalized to every node $n$ in the network by 
substituting $q_{k,u}$ with the local residual popularity 
$q_k(n)$ at node $n$ and with  $\E[\RTT_{k,r,u}]$
the local virtual residual round trip time for object $k$ on 
route $r$, denoted by $\E[\VR_{k,r}]$. Hence we set the 
probability to store an object $k$ at a given node $n$, 
proportional to the popularity and latency locally observed 
at noted $n$. It is left to future work to prove that this 
distributed heuristic is actually optimal.
The intuition behind eq.(\ref{eq:objective}) is that user $u$ 
downloads an object $k$ from a remote path $r$ inversely 
proportional to its popularity and retrieval latency.
A globally optimal strategy performed in each node would 
heuristically prefer to cache locally popular content and 
with high retrieval latency.

In this paper we design an heuristic based on the following 
criterion. Thus our general formulation of the probability 
that the caching decision $d_i$ is true i.e. the probability 
to cache the $i^{th}$ requested object considering all encountered object retrieval latencies, is

\begin{equation}
\mathbb{P} \left[ d_i = true \right] \propto \min{\left(\frac{\left( \Delta T_i \right)^\beta }{\left( f\left((\Delta T_j)_{j = 1,2,\dots,i}\right) \right)^\gamma }, 1\right)} 
\end{equation}

where  $\Delta T_j$ is the retrieval latency of the $j^{th}$ 
requested object, $f$ might be, for example, either a mean, 
the median or a maximum function, $\beta$ and $\gamma$ are 
intensity parameters, $\propto$ means  ``is proportional 
to''. The object retrieval latency and the probability of 
caching it are, hereby, made proportional.  Note that the 
caching decision may cumulatively depend on another fixed or 
dynamic factors (such as the outcome of another random 
experiment).
In the following section we analyze the performance of the 
proposed caching system and compare it to mechanisms existing 
in the literature under a dynamic workload.

\section{Model}
\label{sec:system_model}
The dynamics of the system are complex to capture in a simple 
model due to the tight coupling between delivery performance 
and caching functions: delivery performance is certainly 
affected by network conditions, while clearly network load is 
a result of caching performance and vice-versa. 

In this section, we first introduce modeling assumptions 
(Sec.\ref{sec:ass}), then proceed in two steps: (i) we tackle 
the single cache case, developing analytically some performance bounds, (ii) we leverage such analysis to 
provide an insight on the network of caches case.

\subsection{Assumptions}\label{sec:ass}
The purpose of this model is to identify the added value of the latency-aware stochastic 
decision in outperforming existing alternatives.
In this context, we consider the smallest set of assumptions to have a simple and feasible 
analytical representation.

\begin{itemize}
\item{Zipf-like popularity:}
We assume that object popularity follows a generalized Zipf 
law. Thus let $q(k)$ be the popularity object with rank $k$:
$ q(k) = \frac{c}{k^\alpha}$
with  $\frac{1}{c} = \sum \limits_i i^{-\alpha}$ and 
$\alpha > 0$. This assumption is widely accepted in the 
literature \cite{breslau:web} 
\cite{Mitra:2011:CWV:1961659.1961662}.

\item{Poisson requests:}
We assume that clients request objects according to a Poisson 
distribution with rate $\lambda$, similarly to 
\cite{Carofiglio:2013:PBS:2542828.2542992}

\item{Independent Reference Model:}
Temporal correlation between object requests, though 
neglected here like in 
\cite{Starobinski:2001:PMW:570289.570293} and 
\cite{Fricker:2012:VAA:2414276.2414286}, is foreseen in 
future extensions of this work.

\item{LRU replacement policy:} We focus on the widely adopted 
LRU replacement policy whose common implementation consists 
in moving the most recently served object to 
the front of a list. This allows to study Move-To-Front 
algorithm as an LRU scheme 
\cite{Jelenkovic:2004:OLC:1024662.1024670}.

\item{Same object size:} For the sake of simplicity, we 
assume that, like in \cite{Che:2006:HWC:2312147.2313846}, all 
retrieved objects have the same size. The model will later be 
improved to encompass more fine-grained features such as 
variable object size. We aim to calculate two metrics that, 
we think, give an insight  of a caching system asymptotic 
behavior: the steady-state miss probability  and mean 
delivery time per popularity rank.
\end{itemize}

Refer to Table \ref{tab:symbols} the notation used throughout 
the paper.

\begin{table}[htb!]
\begin{footnotesize}
\centering
\begin{tabular}{|l||p{7cm}|}
\hline
$x$ & Local cache size in number of objects\\
\hline
$\lambda_k$ & Request rate of rank-$k$ objects. Under Poisson arrivals, $\lambda_k = \lambda q(k)$\\
\hline
$\varphi_{k,\tau}$ & Probability of receiving at least one request for a rank-$k$ object during $\tau$ seconds\\
\hline
$\pi_{k,t}$ & Local cache miss probability for rank-$k$ objects at time $t$\\
\hline
$p_{k,t}$ & Probability of a positive caching decision for object rank $k$ at time $t$\\
\hline
$\tau_x$ & Characteristic Time threshold for filling a cache of size $x$\\
\hline
\end{tabular}
\caption{Notation.}\label{tab:symbols}
\end{footnotesize}
\end{table}

\subsection{LAC in the single cache model} \label{sec:single_cache}
In this section, we analyze the latency-aware mechanism proposed in this paper
by computing its performance, expressed in terms of the cache miss probability.
The analysis starts from the computation of LAC steady-state per object $k$
miss probability, $\pi_k$, $\forall k$. LAC is referred to as $p_i^{asym}$-LRU 
as opposed to systems where either the insertion is determined by a constant 
probability $p$ or insertion/replacement operations are symmetrically driven 
by the same probability. Indeed, recall that LAC asymmetry stems from the fact 
that the insertion is probabilistically determined on a per-object basis by the 
monitored residual latency, while using LRU replacement.

\begin{proposition} \label{prop:MTF_arrival_process}
In a LRU cache with insertion probability $p_{k,t}$, 
the move to front probability at time $t$, during the time window $\tau$, 
for object $k$ is given by
\begin{align}
F_k(t, \tau) \triangleq \left((1 - \pi_{k,t}) + \pi_{k,t}p_{k,t} \right)\varphi_{k,\tau} \\
= \left(1 -   (1-p_{k,t})\pi_{k,t} \right)\varphi_{k,\tau}
\end{align}
being  $\varphi_{k,\tau}$ be the probability of receiving at least one request for a 
rank-$k$ object during $\tau$ seconds.
\end{proposition}

The characteristic time (``Che'') approximation 
\cite{Che:2006:HWC:2312147.2313846} states that for a 
sufficiently large cache, the object eviction time is well 
approximated by a unique constant $\tau_x$, being $x$ the 
cache size.  Under this approximation, hence, the miss 
process for a cache under stochastic caching decision, 
$F_k(t, \tau_x) = 1 - \mathbb{P}_{k,t}[ \sharp MTF > x] = (1 
- (1 - p_{k,t})\pi_{k,t}) \varphi_{k,\tau_x}$.
$\sharp MTF$ denotes the number of distinct objects moved to 
the cache front. Upon the assumption that every object gets 
eventually cached at least once over time.
Under this approximation $F_k(t, \tau_x) \approx 1-\pi_{k,t}$ 
which implies $\pi_{k,t} {\approx} \frac{1 - 
\varphi_{k,\tau_x}}{1 - \varphi_{k,\tau_x}(1 - p_{k,t})}$
and generalizes what obtained in 
\cite{DBLP:journals/corr/MartinaGL13} and 
\cite{Bianchi:2013:CBS:2500098.2500106}.
for any inter-arrival time distributions of the request 
process. If we assume that $p_{k,t}$ and $\varphi_{k,\tau}$  
are both ergodic 

\begin{align}
\mathbb{E}[\pi_k] &=  \int_{0}^{1}{\frac{1 - 
\varphi_{k,\tau_x}}{1 - \varphi_{k,\tau_x}(1 - u)} 
d\mathbb{P}[p_k \leq u]} \nonumber \\
&= 1 - \int_0^1 \mathbb{P}[p_k > u] \dfrac{(1 - 
\varphi_{k,\tau_x}) \varphi_{k,\tau_x}}{(1 - 
\varphi_{k,\tau_x}(1 - u))^2}du
\end{align}
If we restrict to a discrete set of positive caching decision 
probabilities,

\begin{equation}
\label{eq:pi_asym_mean}
\mathbb{E}[\pi_k] = \sum \limits_u \mathbb{P}[p_k = u] 
\frac{1 - \varphi_{k,\tau_x}}{1 - \varphi_{k,\tau_x}(1 - u)}.
\end{equation}

\begin{equation}
\label{eq:che}
\tau_x \text{ is the root of }  
\sum \limits_k (1 - \pi_{k}) = x
\end{equation}

That holds from the Che approximation. Note that 
$\varphi_{k,\tau_x} \triangleq 1 - e^{-\lambda_k \tau_x}$ 
under Poisson object arrivals. Note that 
Eq.\eqref{eq:pi_asym_mean} might not be computationally 
tractable. However the following theorem shows that values of 
$p_k$ can be replaced by its mean.

\begin{theorem}
\label{theo:pi_asym_LRU}
If positive caching decision probabilities $p_k$ and 
popularity ranks are deemed independent, and assuming Poisson 
object arrivals,

\begin{equation}
\mathbb{E}[\pi_k] = \frac{e^{-\lambda_k \tau_x}}{ 1 - (1 - 
e^{-\lambda_k \tau_x})(1 - \mathbb{E}[p])}.
\end{equation}
\label{jensen_asym}
\end{theorem}

\begin{proof}
The miss probabilities are a convex function of the caching 
decision probabilities as
$$\frac{\partial^2 \pi_{k,t}}{\partial p^2_{k,t}} \equiv 
\frac{2\left(e^{\lambda_k \mathbb{E}[\tau_x(p)]} - 
1\right)^2}{\left(1 + \big(e^{\lambda_k 
\mathbb{E}[\tau_x(p)]} - 1\big) p_{k,t}\right)^3} \geq 0$$
By Jensen's inequality,
\begin{equation}
\pi_{k} \geq \dfrac{e^{-\lambda_k \mathbb{E}[\tau_x(p)]}}{ 1 - (1 - e^{-\lambda_k \mathbb{E}[\tau_x(p)]})(1 - \mathbb{E}[p])}.
\label{first_jensen}
\end{equation}
A proof that also $\tau_x(p_{k,t})$ is convex  follows by 
computing 
$\frac{\partial^2}{\partial p^2_{k,t}}\tau_x(p_{k,t})$ from the first derivative obtained 
in \cite{DBLP:journals/corr/MartinaGL13} Appendix A, which applies the implicit function 
theorem over  Eq.\eqref{eq:che}. 
$\frac{\partial^2}{\partial p^2_{k,t}}\tau_x(p_{k,t})$ is composed of positive terms expect one 
$\frac{\partial^2}{\partial \tau_x \partial p_{k,t}}I(\tau_x,p_{k,t})$ that is negative $\forall k \leq (\lambda \frac{c}{\log(2)}\tau_x)^{\frac{1}{\alpha}}$ and 
surely positive for any sufficiently large number of 
popularity ranks. Hence, by invoking Jensen's inequality once more,  $\mathbb{E}[\tau_x(p)] \geq \tau_x(\mathbb{E}[p])$.
As $\frac{\partial \pi_{k,t}}{\partial \tau_x} \equiv 
\frac{-\lambda_k p_{k,t}e^{\lambda_k 
\tau_x}}{(1+(e^{\lambda_k \tau_x}-1)p_{k,t})^2} \leq 0$, $ 
\tau_x$ tends to counterbalance inequality 
\eqref{first_jensen}, giving

\begin{equation*}
\pi_k \sim \frac{e^{-\lambda_k \tau_x}}{ 1 - (1 - e^{-\lambda_k \tau_x})(1 - \mathbb{E}[p])}
\end{equation*}
\end{proof}

This result is important because it states that caching based on a random $p$ with values $p_{k,t}$ ends up in a steady-state miss probability similar to the one obtained directly using a constant positive decision probability $\overline{p} = \mathbb{E}[p]$.
Fig.\ref{fig:asym} depicts this miss probability over popularity rank as a function of the decision probability. It gives a first intuition that keeping $\overline{p}$ very small decreases drastically the miss probability of high popularity ranks. The number of beneficiary ranks being limited by the cache size (set to 8 files in this instance).
\begin{figure}[!ht]
\centering
\includegraphics[width=0.4\textwidth]{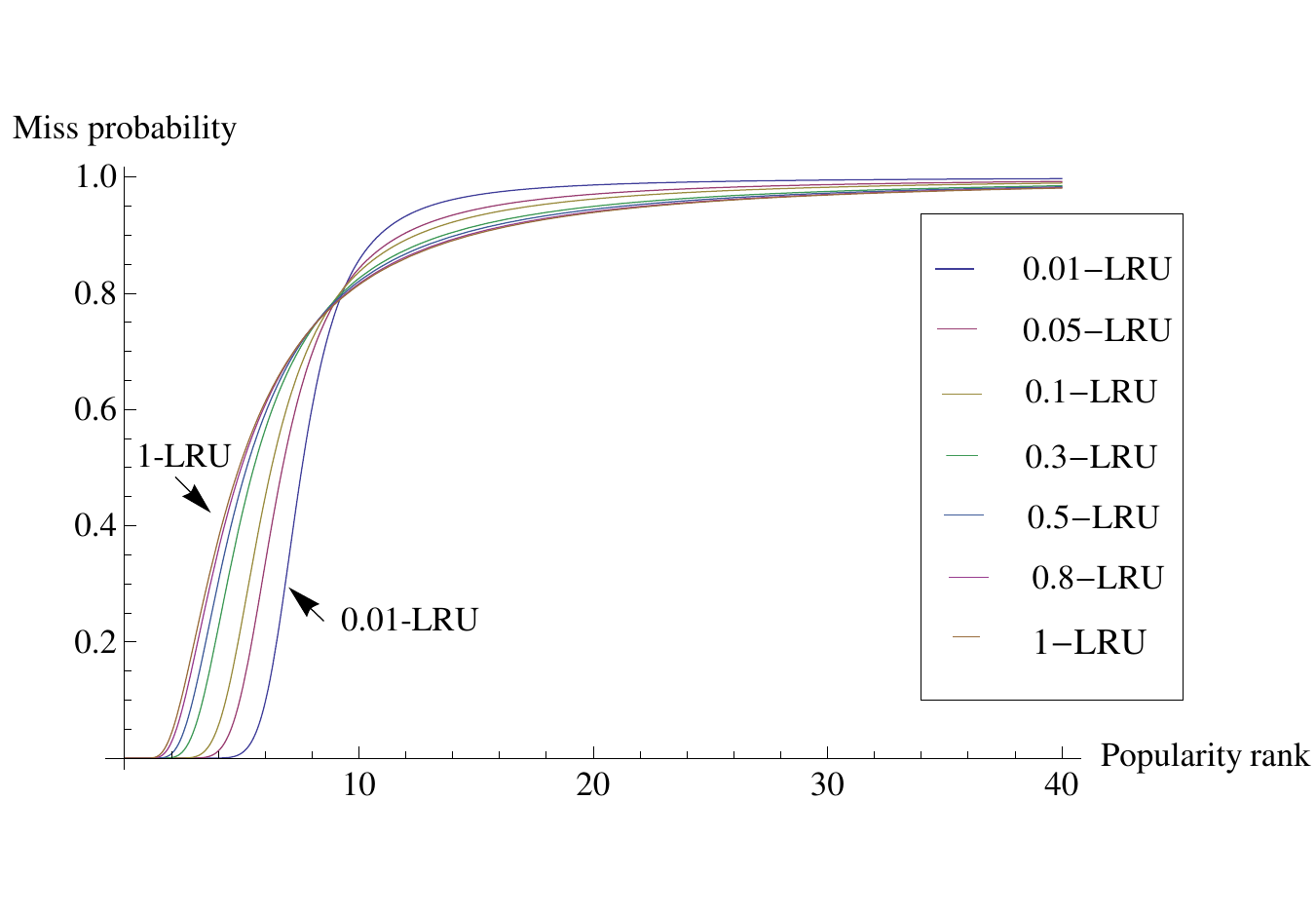}
\caption{$\pi_k$ increases with $\overline{p}$ for the most popular objects.} 
\label{fig:asym} 
\end{figure}
However, the drawback of a constant and small $\overline{p}$ 
is that it postpones considerably the time popular objects 
are first stored in the cache. LRU+LCP suffers from this 
phenomenon because the expected time to enter the cache is
$\frac{1}{\lambda_k \overline{p}}$. Consequently, the overall 
object delivery time converges slowly.

\subsection{$p_i^{sym}$-LRU, an analytical lower bound to $p_i^{asym}$-LRU}
\label{subsec:lower_bound}
Providing a closed-form expression for $p_i^{asym}$-LRU's 
miss probability and its Characteristic Time  $\tau^{asym}_x$
is hard. Instead, we demonstrate its superiority over the 
analytically tractable $p_i^{sym}$-LRU mechanism.
With some loss of generality, $\alpha$ is assumed greater 
than one. Let us consider the symmetric mechanism 
$p_i^{sym}$-LRU where the MTF probabilities are conditioned  
by the same probability $p_{k,t}$. By contrast in 
$p_i^{asym}$-LRU the MTF decision is taken in case of miss 
only.
\begin{theorem} \label{theo:pi_sym_LRU}
$p_i^{sym}$-LRU steady-state miss probability
$\pi^{sym}_{k}  =  \exp\{-\frac{x^{\alpha}}{k^{\alpha}\Gamma (1 - 
\frac{1}{\alpha} )^{\alpha}}\}$
under the assumption that the mean positive caching decision 
probability, $\overline{p}$, is the same for all popularity 
ranks i.e. $\overline{p} = \sum \limits_u \mathbb{P}\left[p_k 
= u \right] u, \forall k$.
\end{theorem}

\begin{proof} 
Let $\sharp k$ denote the number of times a rank-$k$ object 
is moved to the cache front. The mean number of distinct objects moved 
to the front of the LRU cache over time is
\begin{align}
\sum \limits_k \mathbb{E}\left[\mathds{1}_{\{\sharp k > 0\}}\right] &= \sum \limits_k 1 - e^{-\lambda_k \tau \overline{p}}\approx \int^{+\infty}_1(1 - e^{-\lambda \frac{c}{v^{\alpha}} \overline{p}} )dv \nonumber  \\
&= (\lambda \tau c \overline{p})^{\frac{1}{\alpha}}\Gamma \left( 1 - \frac{1}{\alpha} \right)\nonumber 
\end{align}
In virtue of Lemma 5 of \cite{Jelenkovic:2004:OLC:1024662.1024670}.
Hence, the power of alpha-magnified mean number of distinct objects moved to the front of the LRU cache over time 
\begin{align}
x^{\alpha} = \lambda \tau^{sym}_x c \overline{p}\Gamma \left( 1 - \frac{1}{\alpha} \right)^{\alpha} \Rightarrow
\tau^{sym}_x \approx  
\frac{x^{\alpha}} {\lambda c \overline{p} \Gamma 
\left( 1 - \frac{1}{\alpha} \right)^{\alpha}}\nonumber
\end{align}
The rest follow by using the exponential inter-arrival distribution
for an object with rank $k$.
\end{proof}
The closed-form expression of Theorem \ref{theo:pi_sym_LRU} is intrinsically the same as LRU's in \cite{Carofiglio:2013:PBS:2542828.2542992}. This observation yields the next corollary.

\begin{corollary}
If decision probabilities and popularity ranks are deemed independent,
$\pi^{sym}_{k} \overset{L^1}{\rightarrow} \pi^{LRU}_k$
i.e. $p_i^{sym}$-LRU behaves in first-order mean like LRU.
\end{corollary}
$p_i^{asym}$-LRU consequently outperforms $p_i^{sym}$-LRU thanks to its convergence to the Least Frequently Used replacement policy\cite{DBLP:journals/corr/MartinaGL13}. This leads to the next theorem.

\begin{theorem}
\label{theorem_asym_gt_sym}
Let $\eta_{mechanism}$ be the number of most popular objects ``permanently'' accommodated thanks to a caching mechanism,
$\exists \mu \geq 1 : \eta_{p_i^{asym}LRU} \geq \mu \eta_{p_i^{sym}LRU}$
\end{theorem}
i.e. $p_i^{asym}$-LRU allows to accommodate ``permanently'' $\mu$-times more of the most popular objects than $p_i^{sym}$. 
\begin{proof}\textit{(optional)}
Let the miss probabilities of all ``permanently'' stored objects admit a sufficiently small value $\epsilon$ as upper bound.
Then, 
$\eta_{p_i^{sym}LRU} = \frac{x}{\Gamma(1 - \frac{1}{\alpha}) (-\log \epsilon)^\frac{1}{\alpha}}$
and 
$\eta_{p_i^{asym}LRU} = \left(\frac{ \lambda c \tau^{asym}_x}{\log(1 + \frac{1}{\mathbb{E}[p]}(\frac{1}{\epsilon} - 1))}\right)^\frac{1}{\alpha}$.
Since a first-order Taylor series expansion of $\epsilon$ for $p_i^{asym}$-LRU, when $\mathbb{E}[p] \to 0$, yields
$\eta_{p_i^{asym}LRU} \underset{\mathbb{E}[p] \to 0}{\sim} \frac{x}{\Gamma(1 - \frac{1}{\alpha}) \mathbb{E}[p] \log(1 + \frac{1}{\mathbb{E}[p]}(\frac{1}{\epsilon} - 1))^\frac{1}{\alpha}}$,

$\lim \limits_{\mathbb{E}[p] \to 0} \frac{\eta_{p_i^{asym}LRU}}{\eta_{p_i^{sym}LRU}} \geq  (-\log \epsilon)^\frac{1}{\alpha} > 1$

\end{proof}
Let $LAasym \triangleq$ LRU equipped for latency-aware stochastic caching decision (presented in this paper) 
and $LAsym \triangleq$  LRU modified for latency-aware stochastic MTF decision (Starobinski-Tse-Jelenkovi\'c-Radovanovi\'c's).
\begin{corollary}
\label{corollary_asym_gt_sym}
As a mere special case of Theorem \ref{theorem_asym_gt_sym}, 

$\exists \mu \geq 1 : \eta_{p_i^{LAasym}} \geq \mu \eta_{p_i^{LAsym}}$.
\end{corollary}

This typically means that the performance of LRU caches equipped with latency-aware stochastic caching decision can exceed beyond a given factor $\mu$ that of $p_i^{sym}$-LRU, then LRU studied analytically and extensively in previous works \cite{Carofiglio:2013:PBS:2542828.2542992}.

\subsection{Network of caches}
\label{subsec:network_of_caches}

The analytical characterization of the dynamics of a network of caches, even in a broader scope than ICN, is an active research topic \cite{DBLP:journals/icl/Blefari-MelazziBCD14} \cite{Rosensweig:2010:AMG:1833515.1833684}\cite{journals/corr/JinYKSYHL13} and some closed-form results have been presented, but only for networks of LRU caches\cite{Carofiglio:2013:PBS:2542828.2542992}. 

Leaving for future work a thorough analytical characterization of network dynamics under LAC, we explain here the entanglement between latency-aware caching and network performance we need to take into account.

Let focus on a single path, where in-network caching is enabled at each node.
As in \cite{ITC,Carofiglio:2013:PBS:2542828.2542992}, we may denote with $\VR_k$, the Virtual Round Trip Time (VRTT) for any packet of object $k$ and at a given user which we assume to be the first node of this path towards the repository. $\VR_k$ is defined as the weighted sum of user-to-node $i$ round trip time, $R(i)$ times the probability for node $i$ to be the first hitting cache for the request sent by the user, 
$$
\VR_k=\sum_i R(i) \prod_{j<i}p_k(j)(1-p_k(i))
$$
$p_k(i)$ being the miss probability for packets of object $k$ at node $i$.
Now, at every intermediate node $l$ along the path, the measured residual latency can be defined as:
$RVRTT_k(l) = \sum_{i\ge l} R(i) \prod_{j<i}p_k(j)(1-p_k(i))$
Over time, the expectation of the Residual Virtual Round Trip Time for a rank-$k$ 
object at node $n$, $\mathbb{E}[RVRTT_k(n)]$ represent the mean cost of all routes 
to a permanent copy of the object departing from $n$. 

Hence, for the caching node $n$, the probability of a positive LAC decision for rank $k$ at each discrete decision instant $t$ results to be proportional to the monitored average Residual Round Trip Time, \\
$p_{k,t}(n) \propto \min (\frac{(\mbox{RVRTT}_{k,t}(n) )^{\beta}}{\left(f^{(n)}(t-1)\right)^{\gamma}} , 1) \text{, }t > 1$.
Beyond the normalization of the probability to $1$, the specific function we have selected accounts for a normalization of the monitored metric over a function,  $f^{(n)}$ which is meant to indicate the overall latency cache $n$ welcomes. 
We have had successful experience with the instance $f^{(n)} \triangleq \E\left(\cdot \right)$, namely the average
%
$f^{(n)}(t) \underset{\substack{
		\\ t \to +\infty}}{\sim} \dfrac{\sum \limits_k  \mathbb{E}[RVRTT_k(n)] \mathbb{E}[p_k(n)]}{ \sum \limits_k \mathbb{E}[p_k(n)]}$.

\section{Performance Evaluation}
\label{sec:performance_evaluation}

We implement and test LAC by means of simulations carried out with the packet-level NDN simulator CCNPL-Sim (\url{https://code.google.com/p/ccnpl-sim/}). We evaluate (i) single cache topologies, then (ii) networks of caches topologies with a single content server on the top and three intermediate layers of caches and a client layer at the bottom.
LAC, our latency-aware LRU denoted as $LAasym$ is tested against two other fully distributed caching management mechanisms: LRU+Leave-Copy-Probabilistically and LRU \cite{1395054}\cite{643}. By fully distributed, we mean mechanisms that do not require the exchange of any specific signalling between caches.

\subsection{Single cache topology}
The following results are achieved in a simulated ICN with a single caching node between the object consumers and the publishing server and with the following parameters. 
\begin{itemize}
	\item Cache sizes are equal to 80KBytes. 
	\item The Poisson process for generating content requests is characterized by a rate of 1 object/s
	\item Objects are requested over a catalog of 20,000 items, according to a  Zipf-like popularity distribution of
	parameter $\alpha = 1.7$. This value of $\alpha$ has been demonstrated realistic \cite{Mitra:2011:CWV:1961659.1961662}. Each file is 10KBytes size.
	\item The two FIFO links from the consumers up to the content publisher have a capacity of 200Kbps and of 30Kbps, respectively. 
	\item Each object conveyed through these links has an average size of 10KBytes, that we also take as fixed packet size.
\end{itemize}
About LAC parameters, caches are equipped for latency-aware stochastic caching decision, with $\beta = \gamma = 5$ to stress the rejection of quickly delivered objects. The function $f$ is the mean latency of all ever-cached objects.
We report the related charts in Fig.\ref{single_topology_results}.
The load $\rho$ of the 30Kbps downlink equals 0.56 when the cache is ruled by LRU, 0.58 under $LAsym$, 0.41 under LRU+LCP and 0.37 under LAC ($LAasym$).

\begin{figure*}[htbp] 
\centering
\subfigure[Miss probability curves validate models.]{\includegraphics[width=0.32\textwidth]{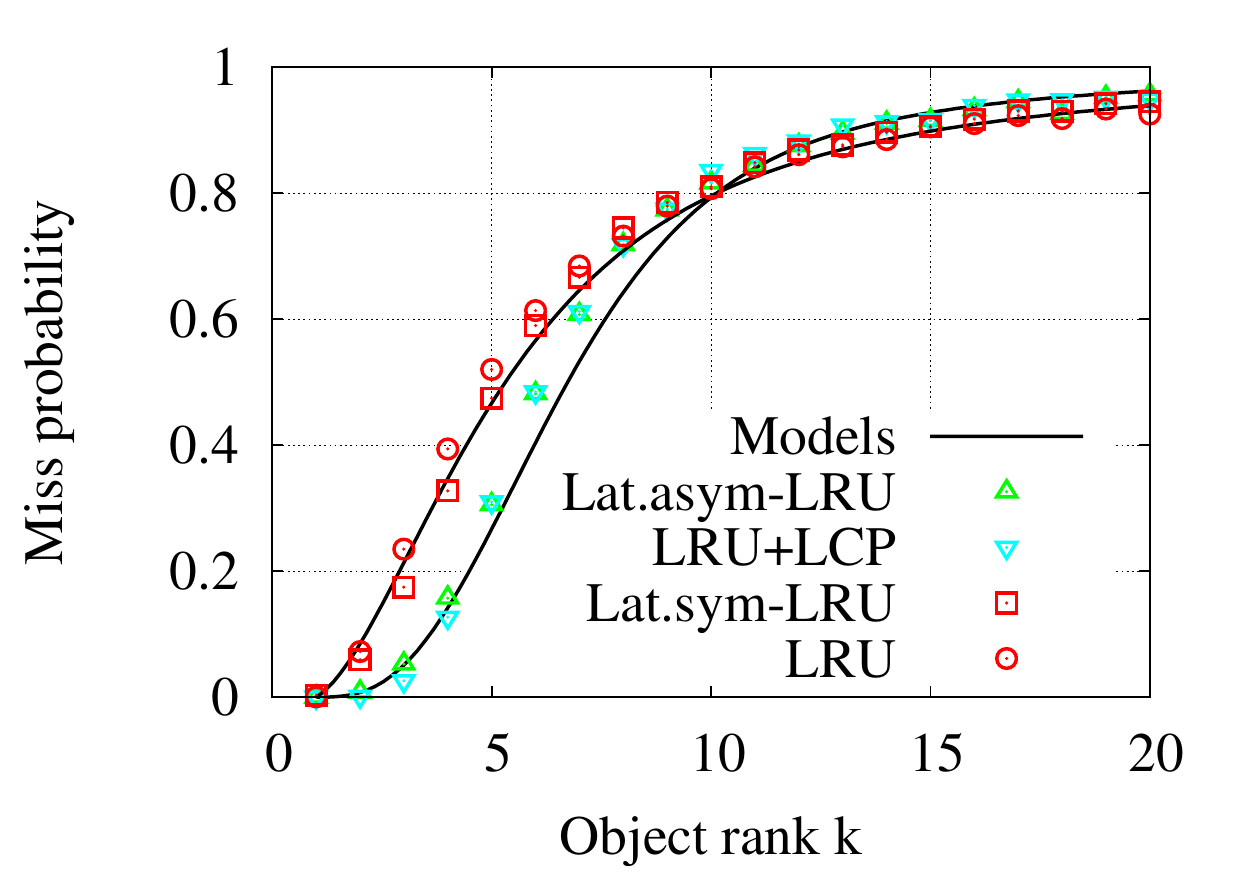}
\label{fig:single_miss_prob}}
\subfigure[Delivery time w.r.t content rank overview (vs LRU)]{\includegraphics[width=0.32\textwidth]{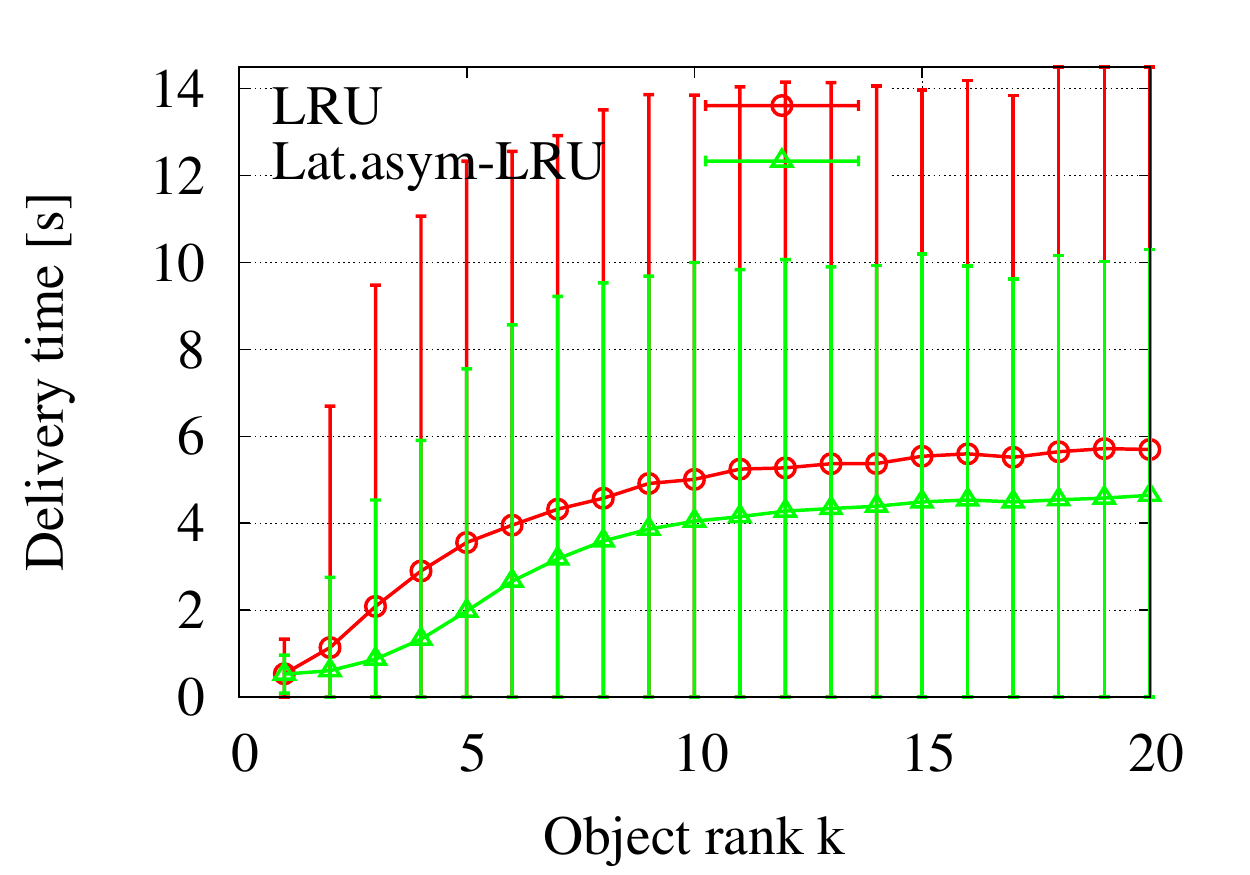}
\label{fig:single_delivery_time_LRU}}
\subfigure[Delivery time w.r.t content rank overview (vs LRU+LCP)]{\includegraphics[width=0.32\textwidth]{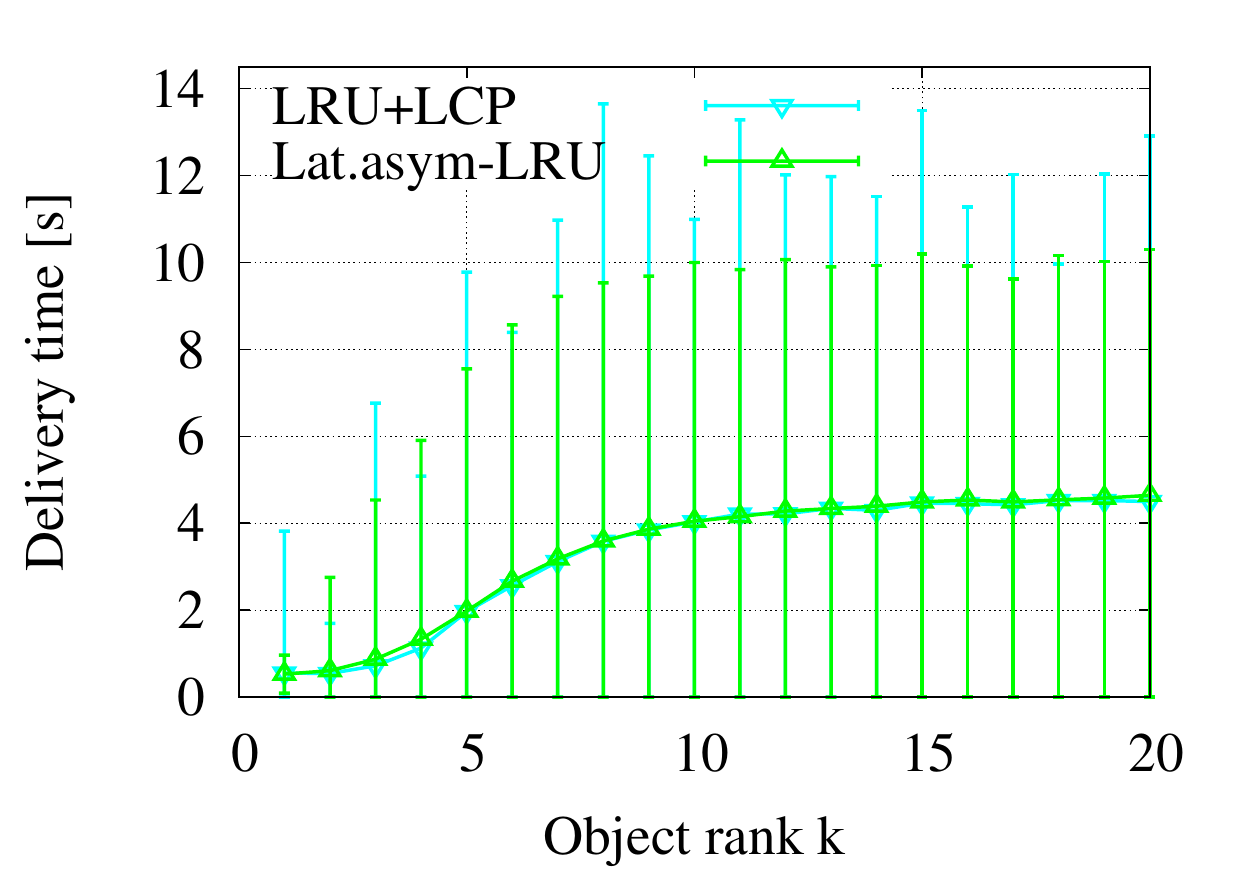}
\label{fig:single_delivery_time_LCP}}
\subfigure[Evolution of the mean delivery time]{\includegraphics[width=0.32\textwidth]{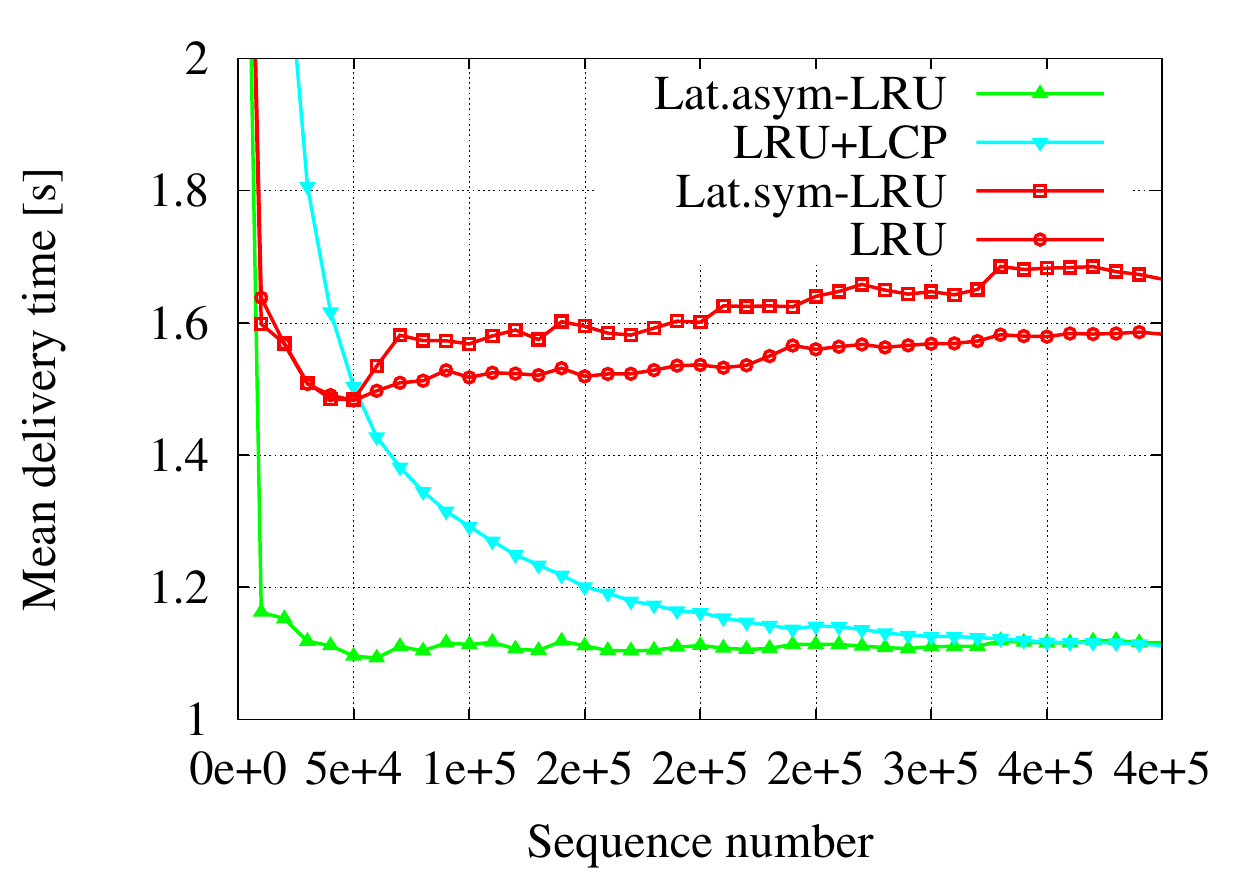}
\label{fig:single_delivery_time_mean_vs_seqnumber}} 
\subfigure[Evolution of the delivery time standard deviation]{\includegraphics[width=0.32\textwidth]{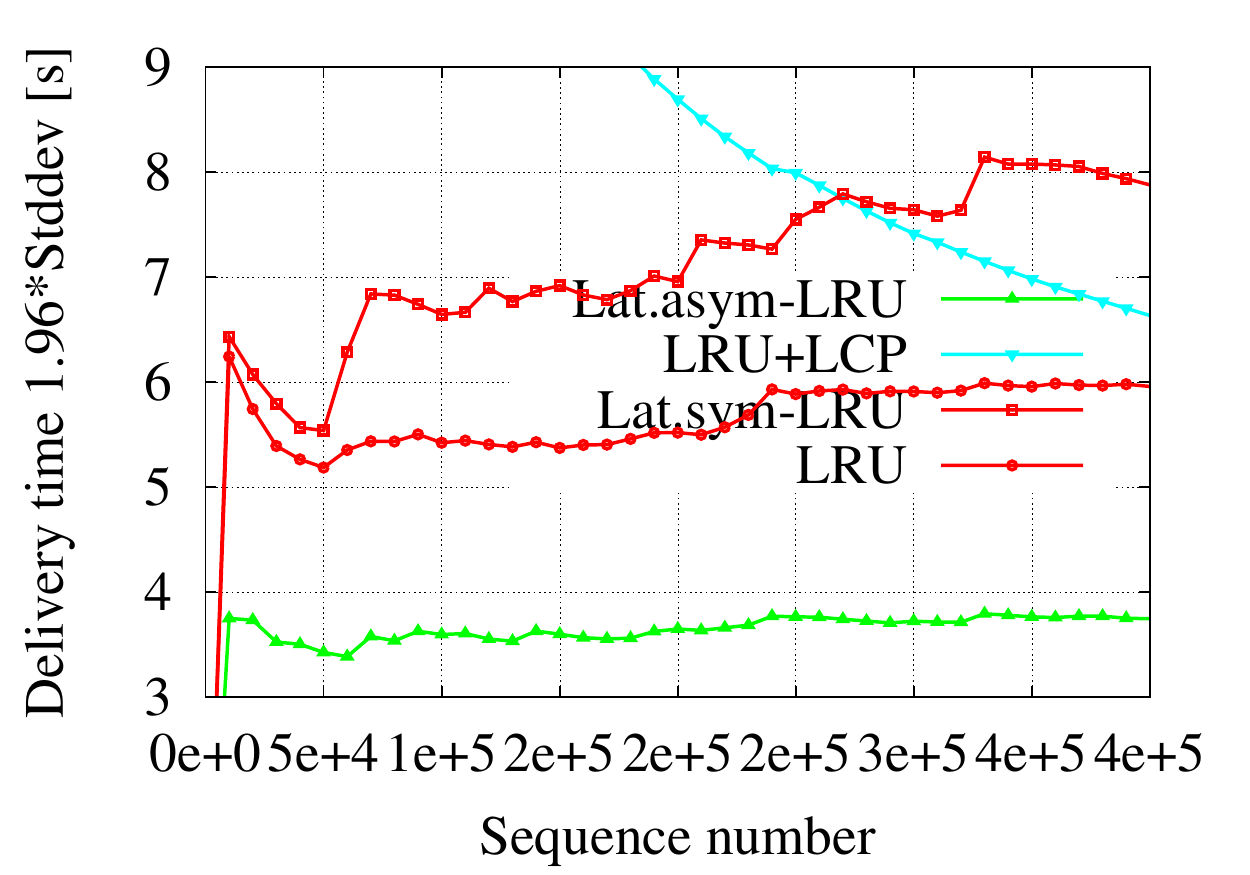}
\label{fig:single_delivery_time_stddev_vs_seqnumber}}
\caption{Single cache topology simulation: LAasym decreases LRU delivery time by 30\% and outperforms LRU+LCP on convergence.}
\label{single_topology_results}
\end{figure*}
A first observation we draw from the plots in Fig.\ref{single_topology_results} is that our LAC proposal, $LAasym$ converges to the same steady state as LRU+LCP, which approximates the optimal LFU behavior.
Note that it this is true in static and hierarchical network of caches with no regeneration (no user requests from intermediate nodes), in general LAC is based on temporal measurements of residual latency, so adapting over time based on the sensed variations in terms of experienced latency. 
Secondly, we observe how much $LAasym$ latency-aware technique reduces both delivery time mean and standard deviation. It is striking to see how quickly they converge, compared to classical LRU+LCP.
The constant decision probability used in LRU+LCP is, indeed, the average of all latency-aware decision probabilities ($p=0.1$) 
and this impacts negatively either the convergence either the system reactivity to temporal variations of latency, as opposed to our LAC proposal. 
Finally, we observe that $LAsym$ and LRU miss probability curves coincide in steady state as predicted in \cite{Jelenkovic:2004:OLC:1024662.1024670}. A symmetric filtering of objects to put in and to remove from the cache has the only effect of slowing down convergence while not modifying the dynamics of the underlying Markov chain. 

\subsection{Network of caches}
\subsubsection{Line topology network}
\begin{figure}[htbp] 
\centering
\includegraphics[width=0.45\textwidth, height=2.5cm]{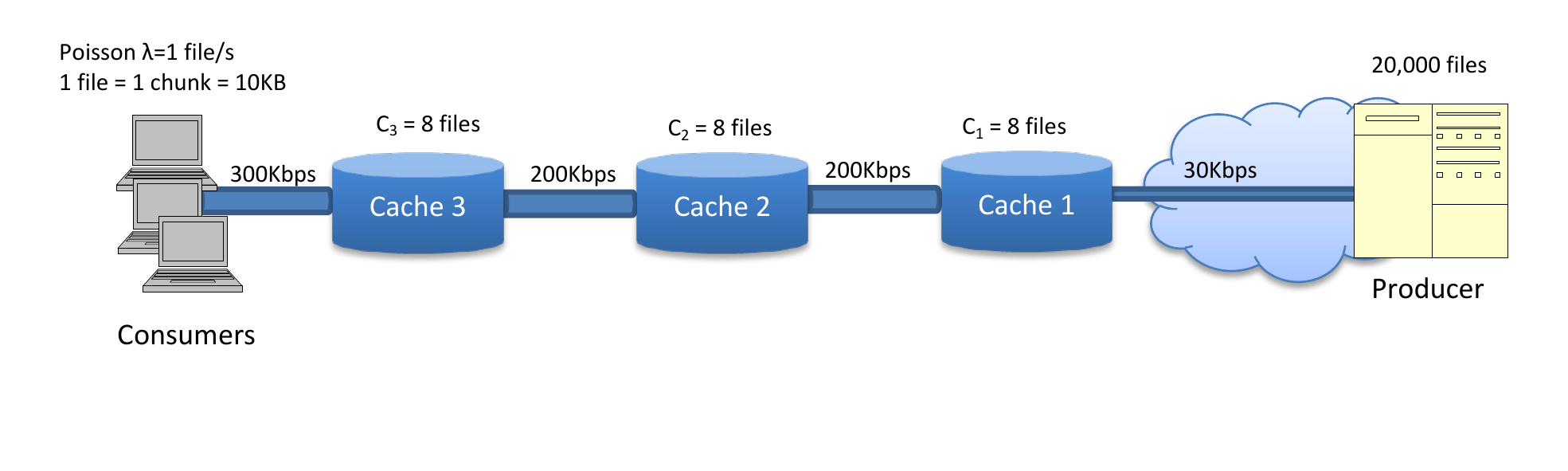}
\vspace{-1cm}
\caption{Simulated line topology.} 
\label{fig:line_topology} 
\end{figure}
We now consider the setting in Fig.\ref{fig:line_topology}, with three caching nodes in-line
between the users and the publishing server. 
The set of parameters we consider is the same as before for what concerns cache sizes, request process and popularity. 
The four links from the consumers up to the
publisher have capacities equal to 300Kbps, 200Kbps, 200Kbps and 30Kbps respectively. 
Under LRU+LCP,  $p = 0.1$ and corresponds to the lowest mean latency-aware caching decision probability, Cache 3's. 
Related results are reported in Fig.\ref{line_topology_results}.
The resulting link load $\rho$ on downlinks from the repository to the users is respectively : (0.5, 0.01, 0.03, 0.27) under LRU,
(0.27, 0.02, 0.02, 0.27) under LRU+LCP and (0.22, 0.04, 0.06, 0.27) under our LAC proposal, $LAasym$. 
\begin{figure*}[ht] 
\centering
\subfigure[LRU and LRU+LCP miss probability]{\includegraphics[width=0.32\textwidth]{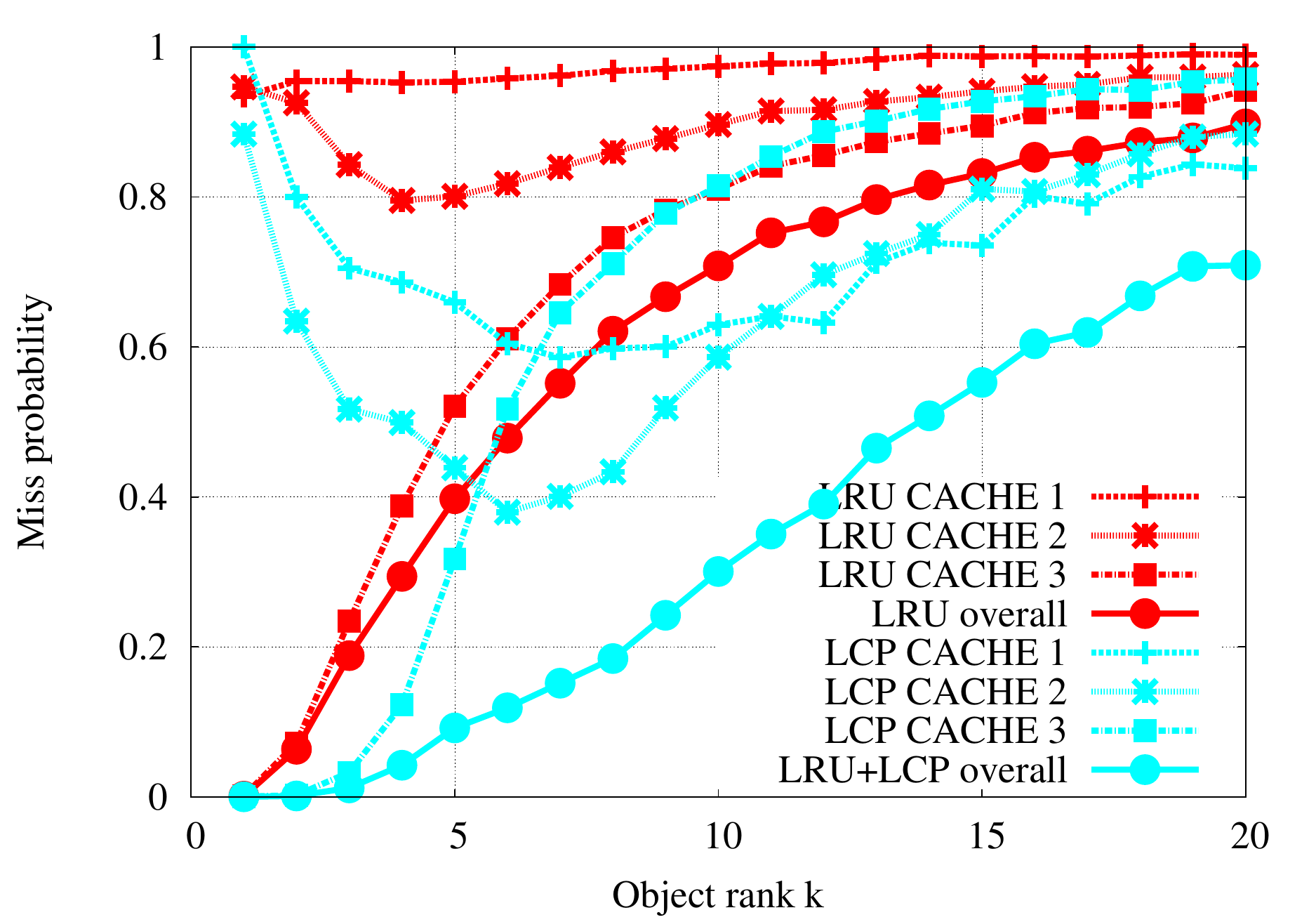}
\label{fig:line_miss_prob_LRU}}
\subfigure[LAasym miss probability]{\includegraphics[width=0.32\textwidth]{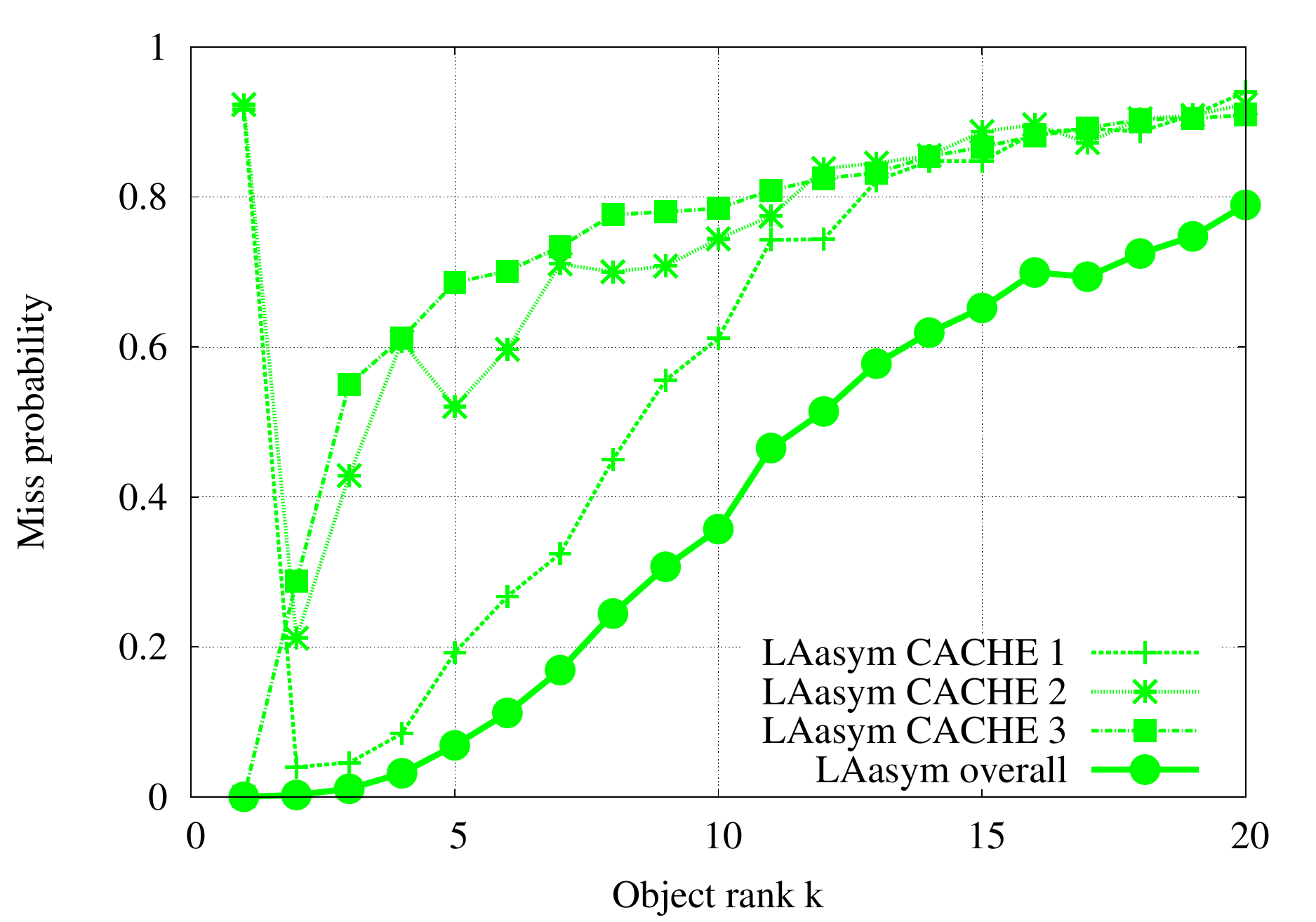}
\label{fig:line_miss_prob_LAasym}}
\subfigure[Delivery time w.r.t content rank overview]{\includegraphics[width=0.32\textwidth]{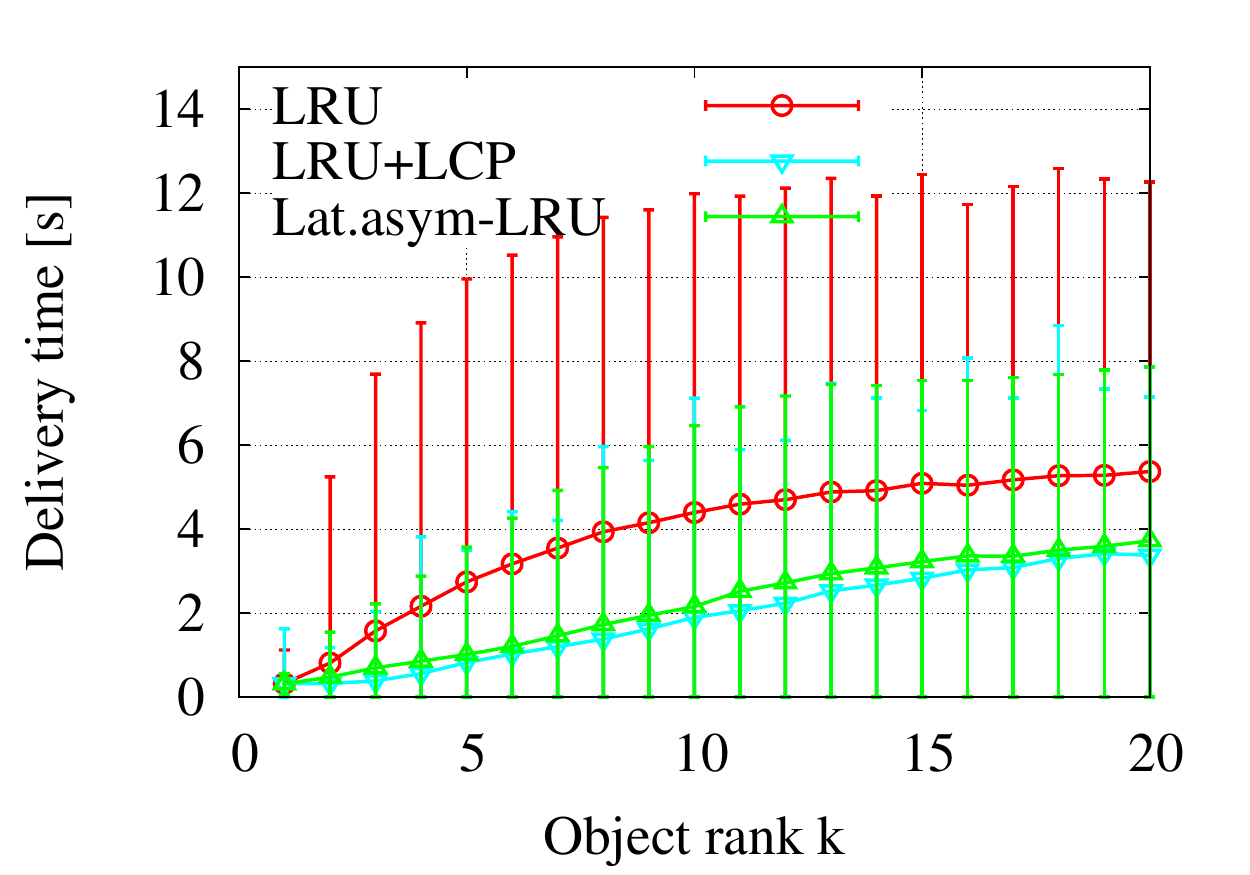}
\label{fig:line_delivery_time}}
\subfigure[Evolution of the mean delivery time]{\includegraphics[width=0.32\textwidth]{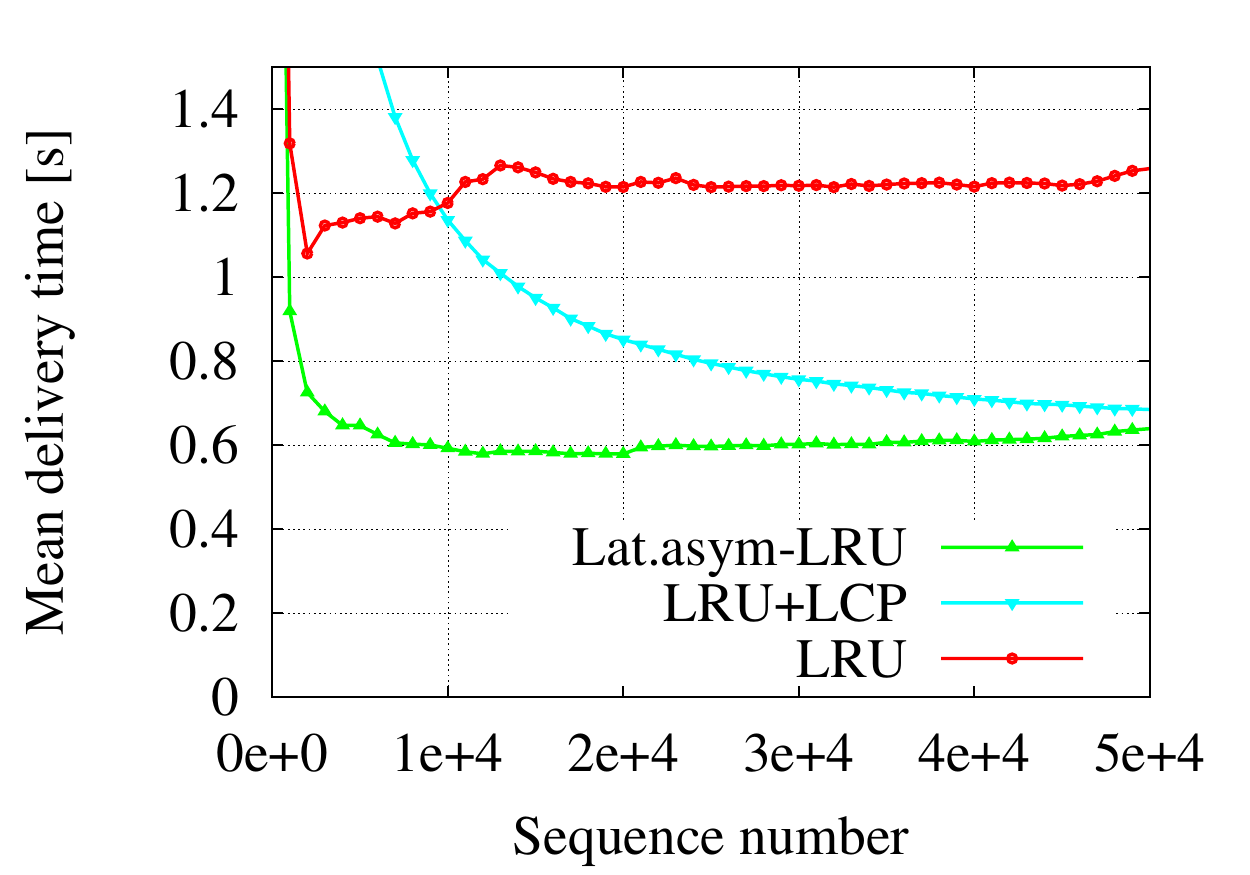}
\label{fig:line_delivery_time_mean_vs_seqnumber}} 
\subfigure[Evolution of the delivery time standard deviation]{\includegraphics[width=0.32\textwidth]{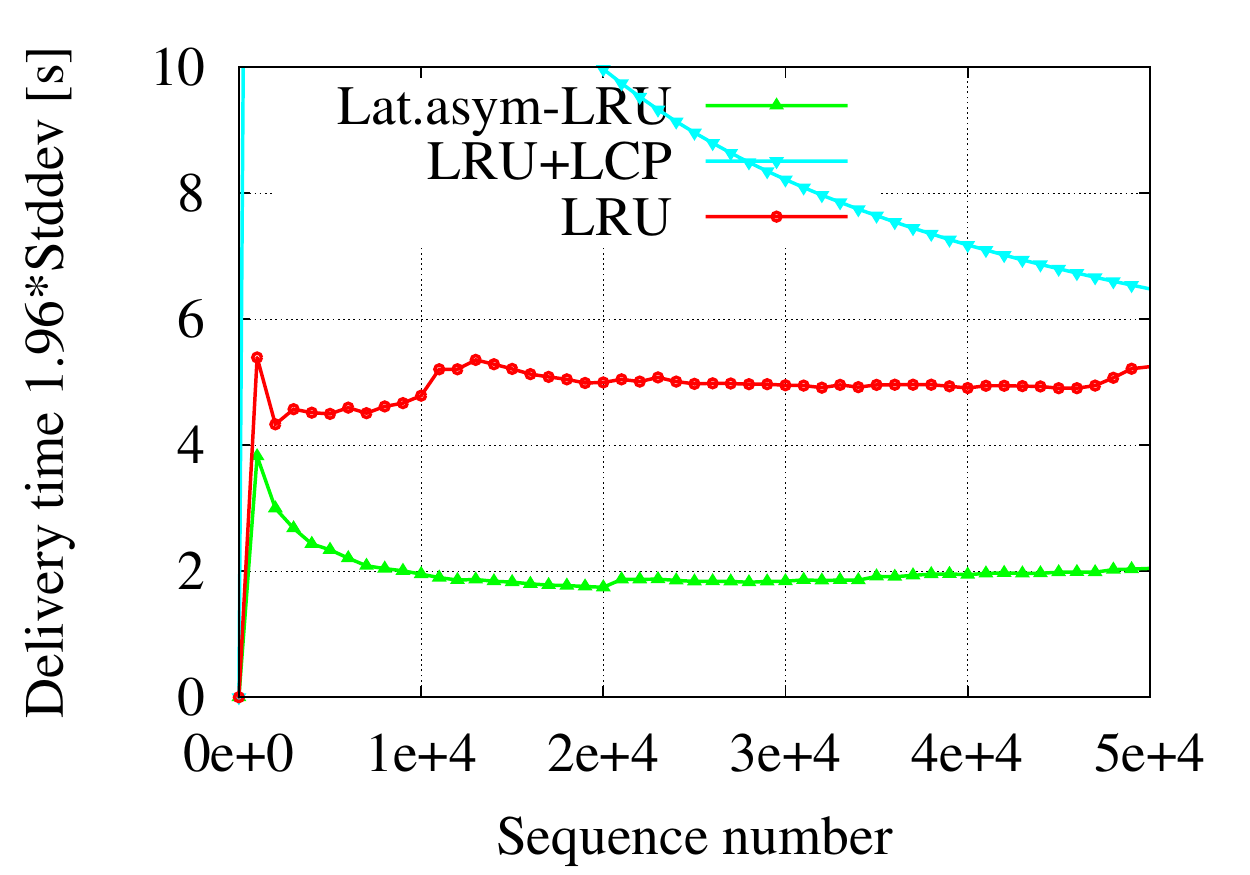}
\label{fig:line_delivery_time_stddev_vs_seqnumber}}
\caption{Line topology simulation: LAasym decreases  LRU delivery time by 50\% and outperforms LRU+LCP on convergence.}
\label{line_topology_results}
\end{figure*} 
Clearly, the expensive traffic to the publisher decreases significantly with $LAasym$, while very little increase can be oberserved on the other links. The tremendous gain in delivery time (50\% of LRU's) can be appreciated in both its first and second moments. Such a delivery time standard deviation decrease plays a central role in stabilizing customers quality of experience.

\subsubsection{Tree topology network}
The next results are those achieved in the ICN setting in Fig.\ref{fig:tree_topology}, spanning a binary tree topology whose
seven caching nodes are spread over three network levels, between the users and
the repository (publishing server). In such configuration, 
\begin{itemize}
	\item Cache sizes are 8MBytes.
	\item Poisson object request rate at the user is 1 object/s.
	\item Object popularity follow a Zipf(1.7) distribution.
	\item Object size is taken equal to 1 MB.
	\item Downlink capacities from the users up to the repository are 
	30Mbps-capable, except the last one toward the repository, which is 9Mbps.
	\item Each packet has an average size of 10KBytes, making every object equal to 100 packets in size.
\end{itemize}

Caches are equipped for LAC decision, with $\beta = \gamma = 3$, with  the function $f$ remaining equal to the mean latency of all ever-cached objects.
Cache 4 is on the first layer (the closest to the consumers), Cache 8 on the second layer and
Cache 10 on the third (the farthest to the users). 
LRU+LCP's $p = 0.03$. That corresponds to $LAasym$'s mean latency-aware caching decision probability. 
\begin{figure}[htbp] 
\centering
\includegraphics[width=0.35\textwidth]{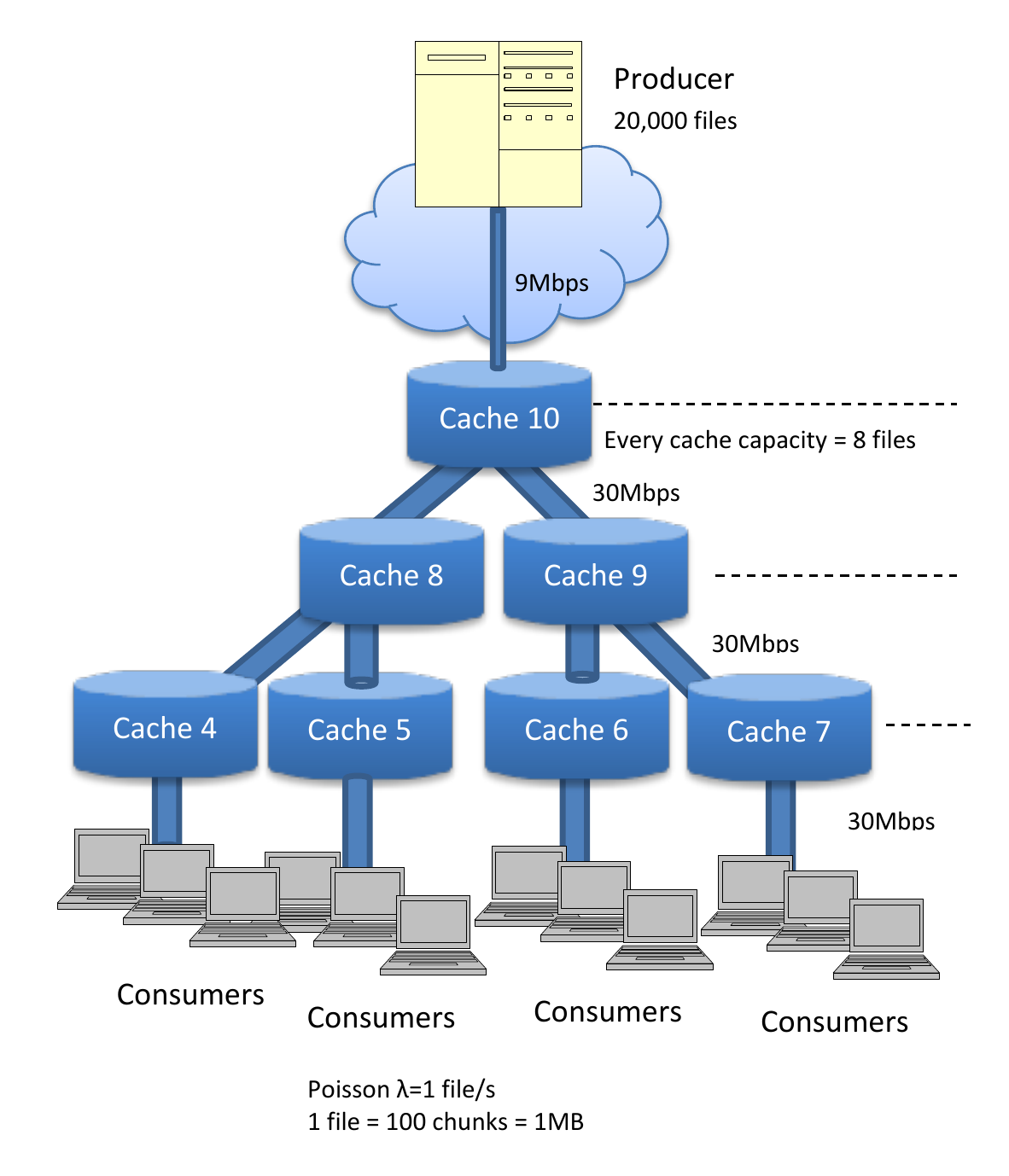}
\caption{Simulated tree topology } 
\label{fig:tree_topology} 
\end{figure}

We report the related charts in Fig.\ref{tree_topology_results}. 
The observed  link load $\rho$ on downlinks from the repository to the users is respectively: (0.7, 0.31, 0.18, 0.6) under LRU, 
(0.7, 0.07, 0.33, 0.6) under LRU+LCP and (0.7, 0.12, 0.23, 0.6) under $LAasym$. Again, our LAC mechanisms allows to lower maximum and average link load over the network. 
 
 \begin{figure*}[htbp] 
 \centering
 \subfigure[LRU and LRU+LCP miss probability]{\includegraphics[width=0.32\textwidth]{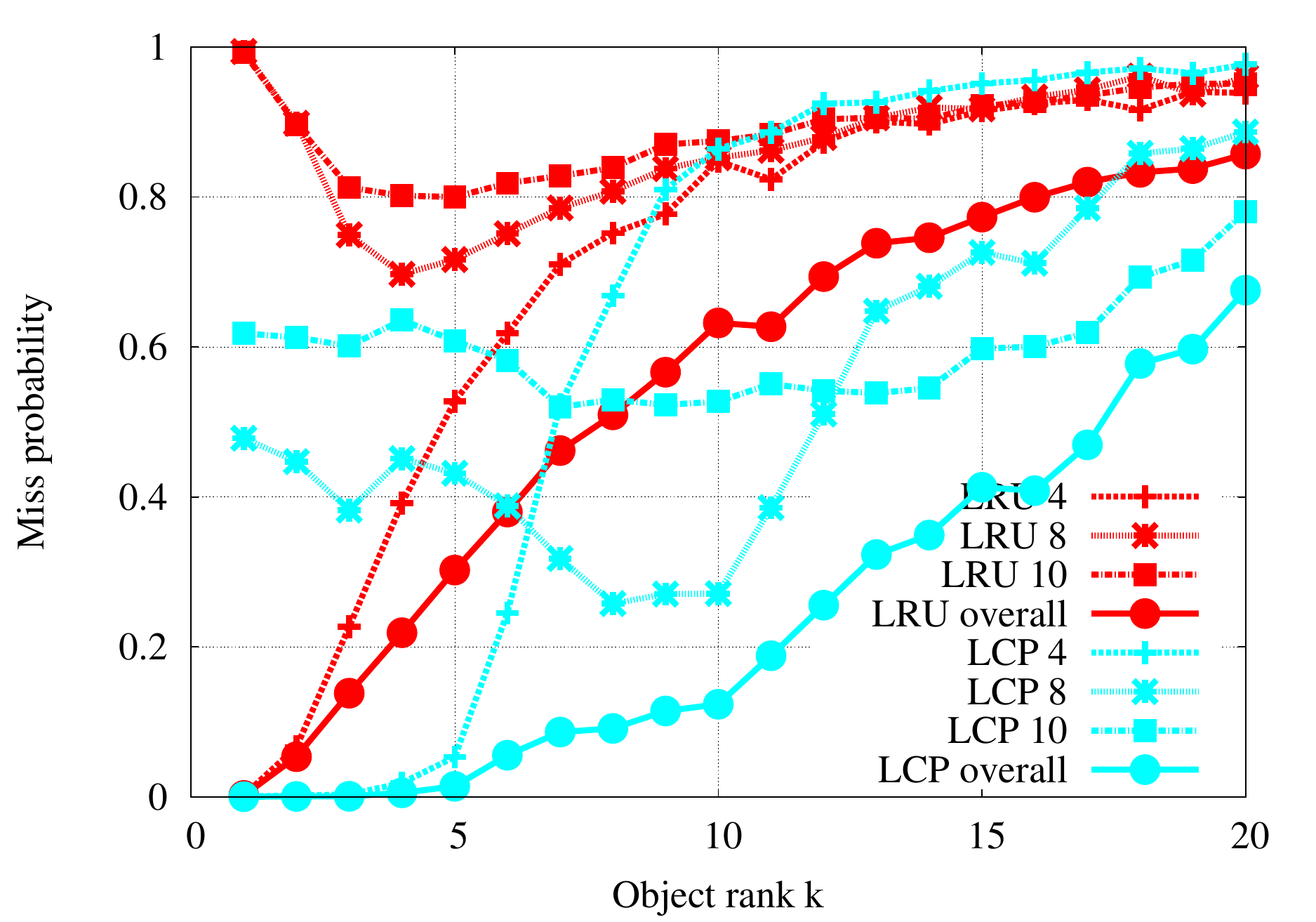}
 \label{fig:tree_miss_prob_LRU}}
 \subfigure[LAasym miss probability]{\includegraphics[width=0.32\textwidth]{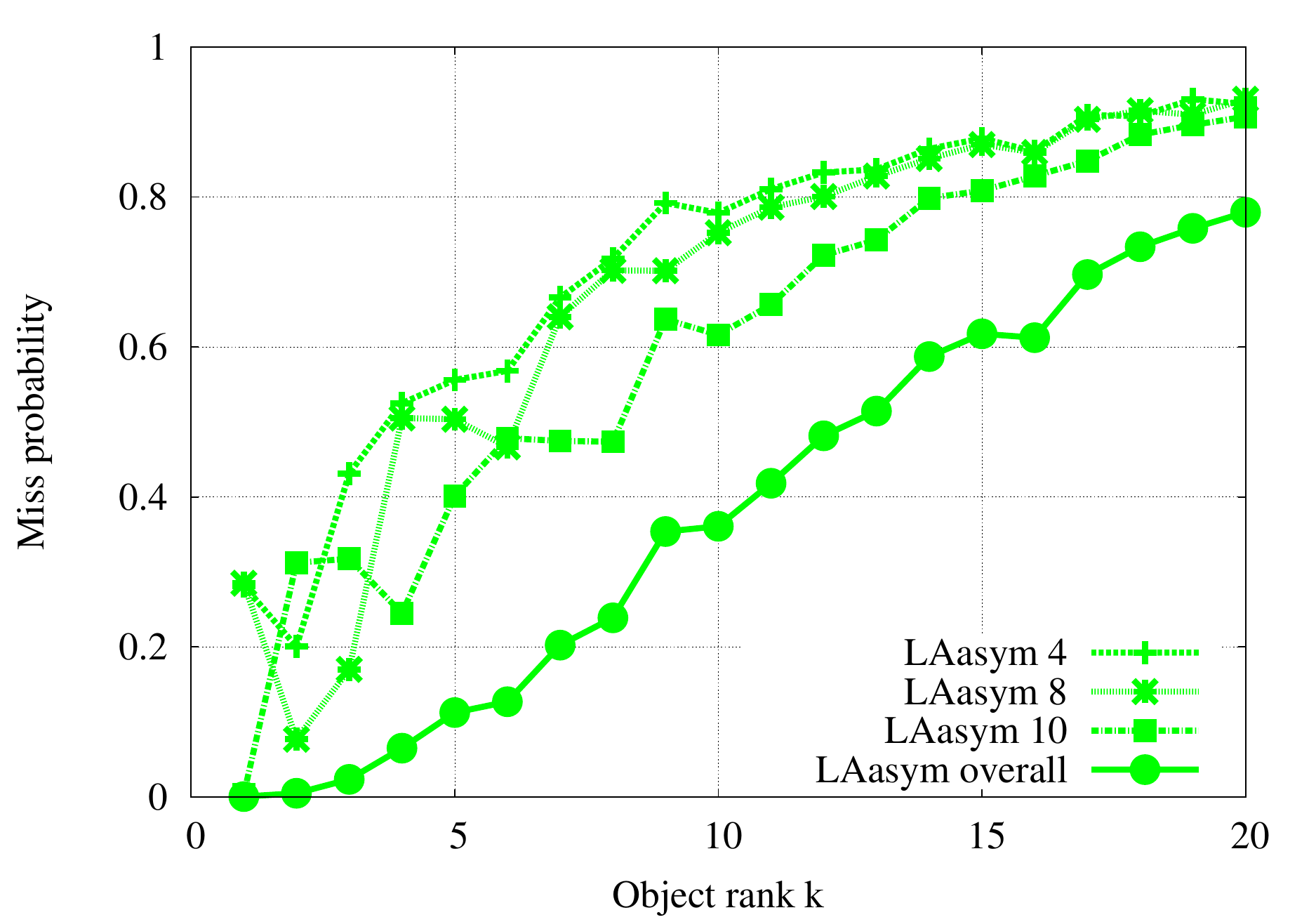}
 \label{fig:tree_miss_prob_LAasym}}
 \subfigure[Delivery time w.r.t content rank overview (vs LRU)]{\includegraphics[width=0.32\textwidth]{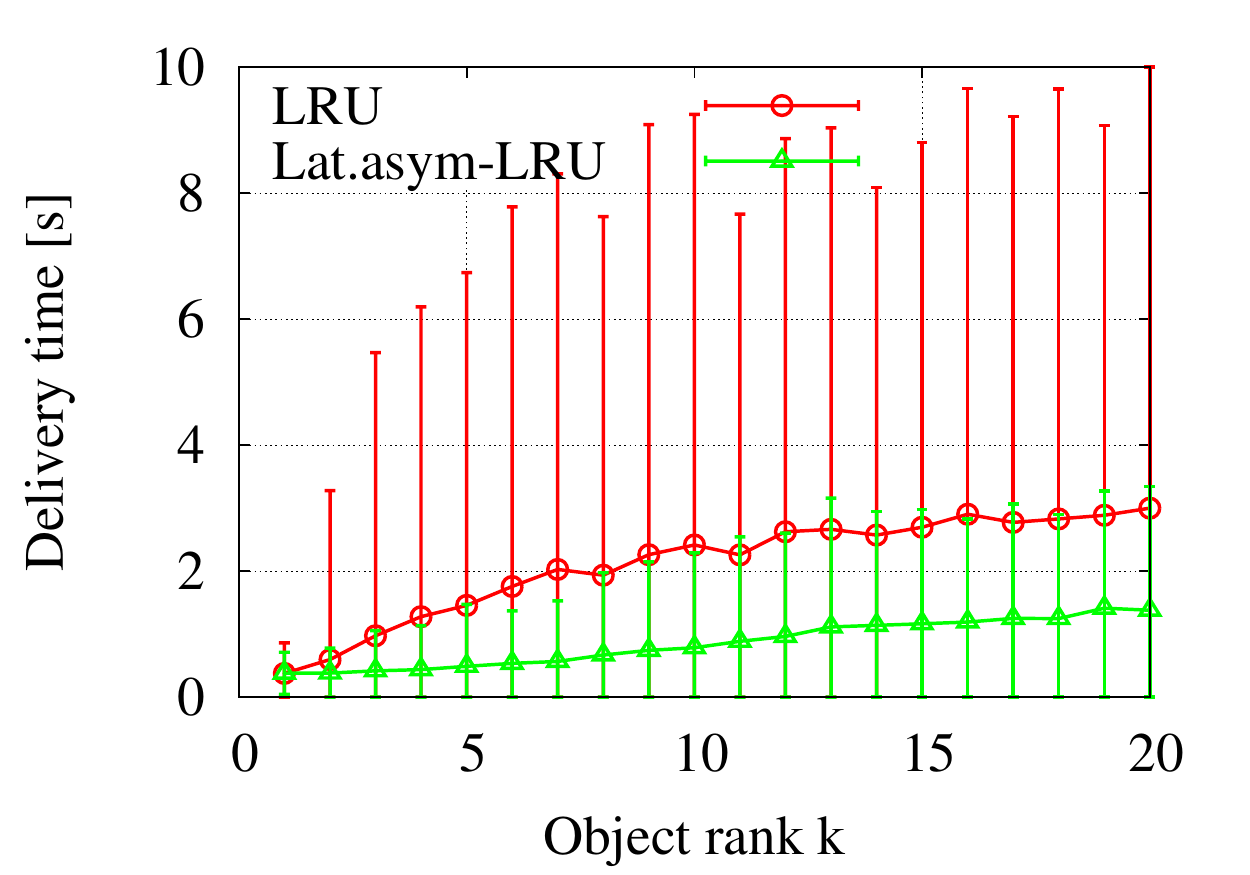}
 \label{fig:tree_delivery_time_LRU}}
 \subfigure[Delivery time w.r.t content rank overview (vs LRU+LCP)]{\includegraphics[width=0.32\textwidth]{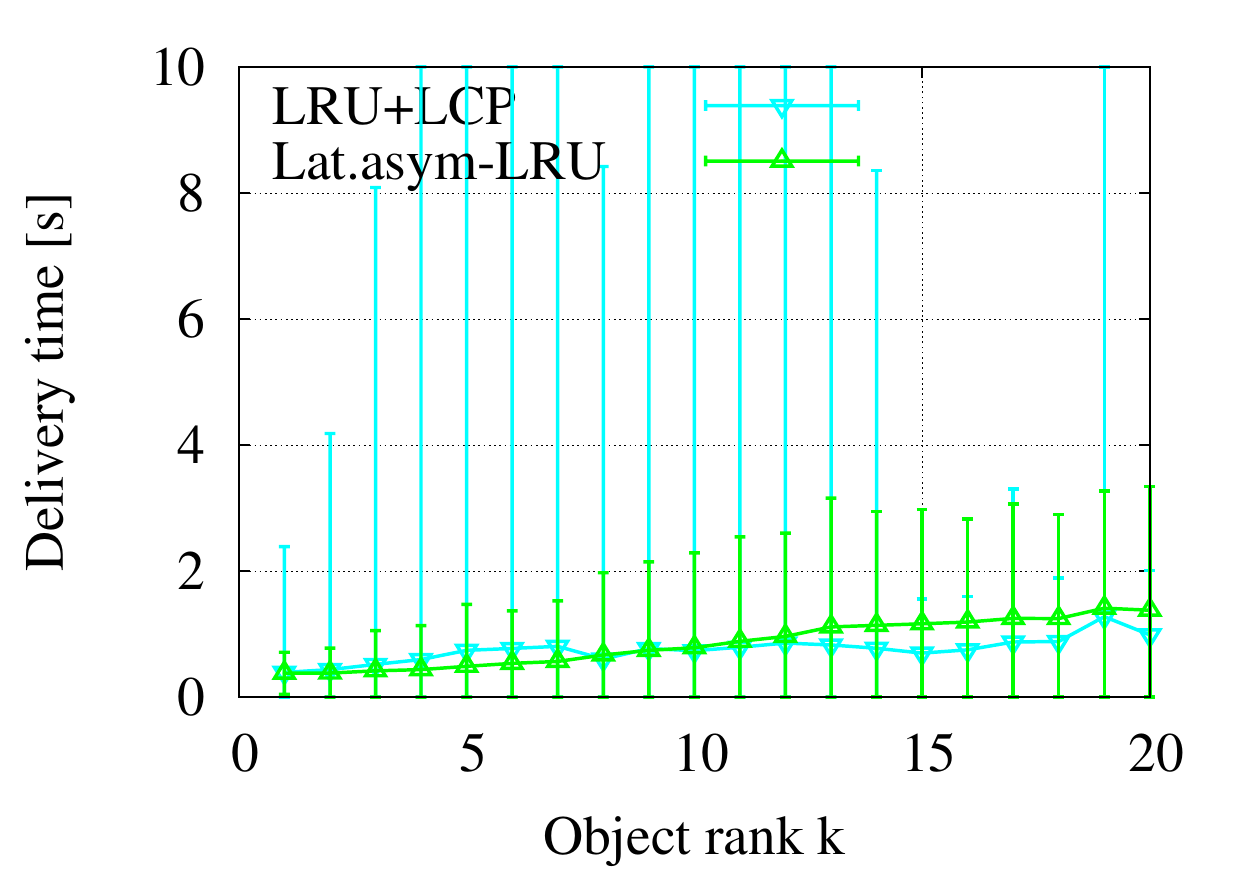}
 \label{fig:tree_delivery_time_LCP}}
 \subfigure[Evolution of the mean delivery time]{\includegraphics[width=0.32\textwidth]{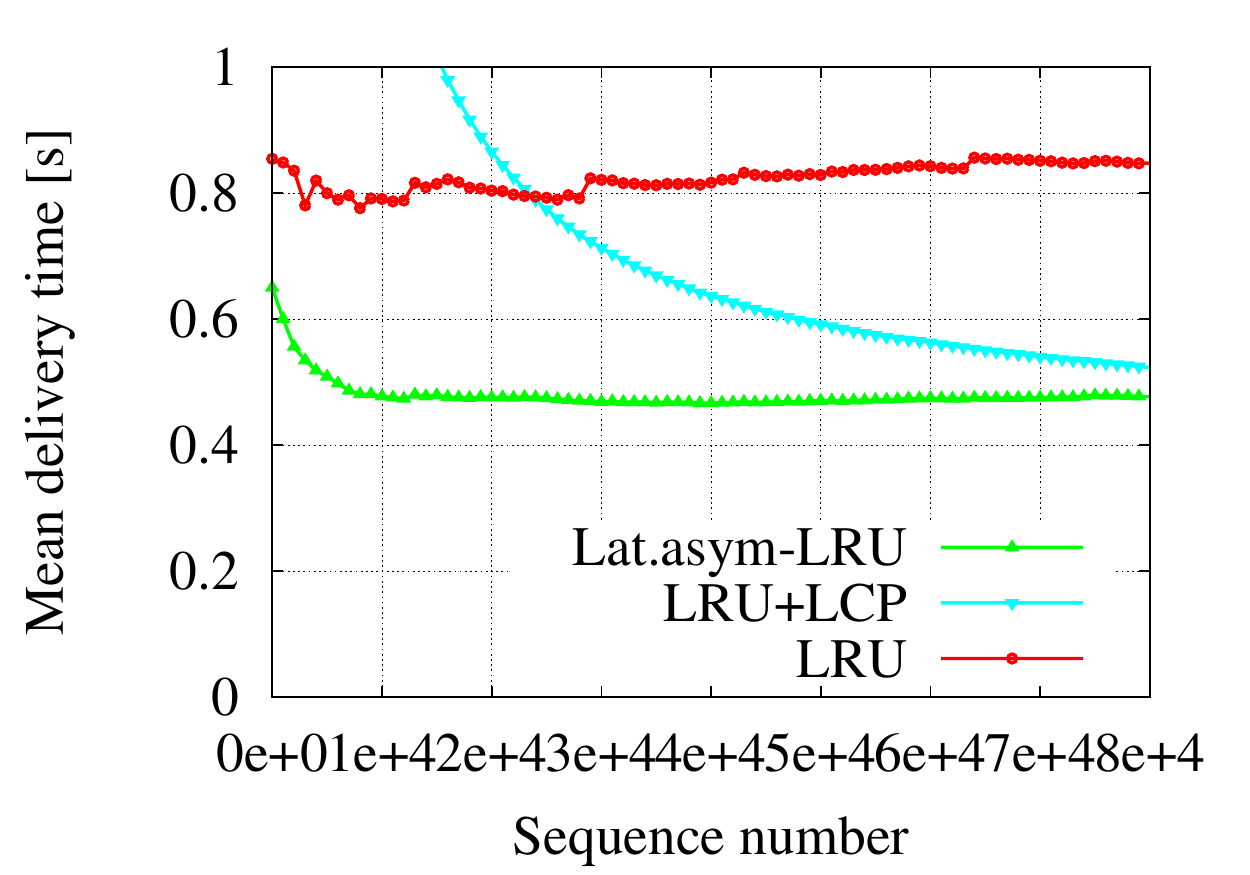}
 \label{fig:tree_delivery_time_mean_vs_seqnumber}} 
 \subfigure[Evolution of the delivery time standard deviation]{\includegraphics[width=0.32\textwidth]{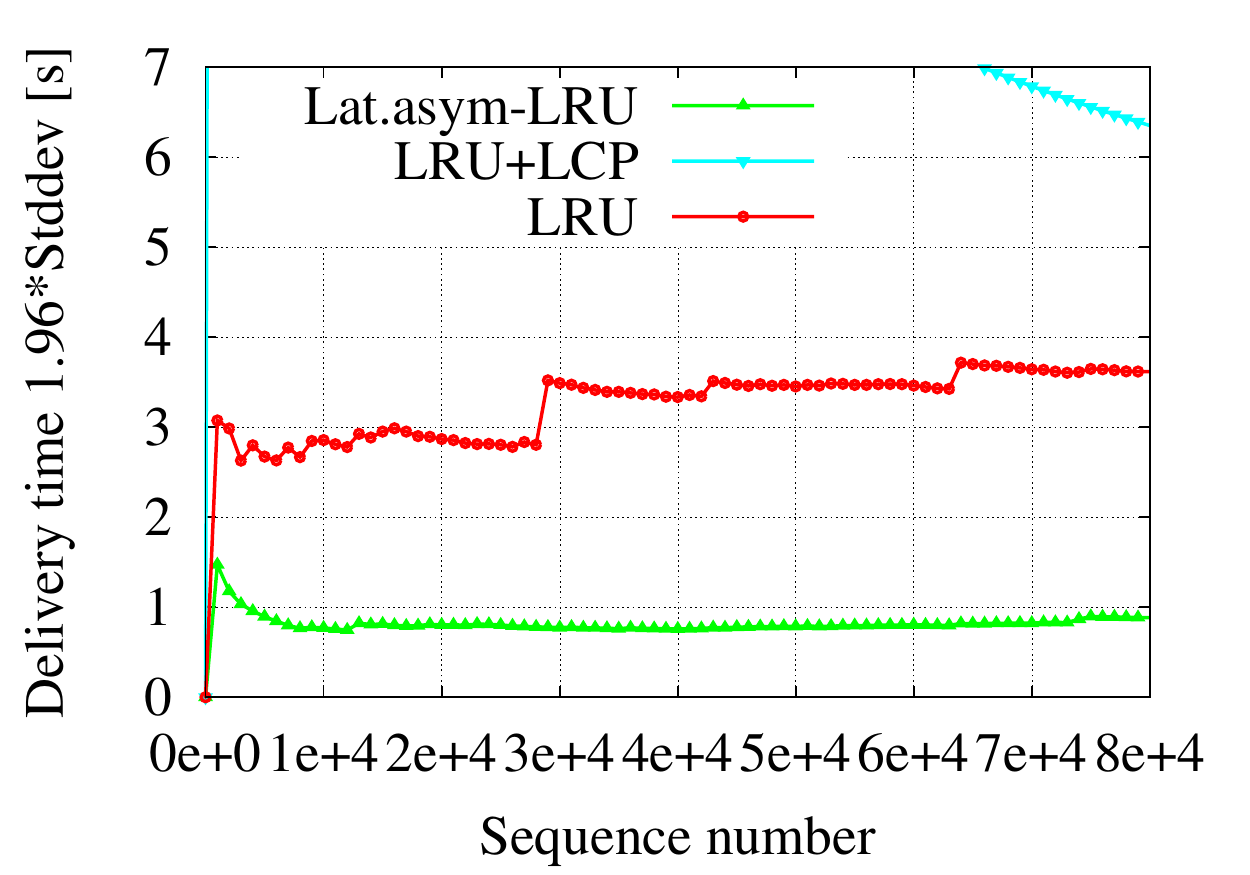}
 \label{fig:tree_delivery_time_stddev_vs_seqnumber}}
 \caption{Tree topology simulation:  LAasym decreases LRU delivery time by 30\% and outperforms LRU+LCP on convergence.}
 \label{tree_topology_results}
 \end{figure*}
 
Finally, we observe as a general rule that implementing $LAasym$ decreases the
overall cache miss probability i.e. the probability that all solicited caches fail to serve the
requested object. It also decreases and stabilizes the overall object delivery time. Indeed, the mean delivery time and the 95\% confidence interval around the average, both decrease by up to 50\%. 

Note also that this overall improvement is not achieved to the detriment of the convergence speed, unlike LRU+LCP. The latter, indeed, exhibits tremendously slow convergence and extremely high delivery time standard deviation.

\section{Conclusion and future work}
\label{sec:conclusion}
In the paper, we showed the benefits of leveraging latency for caching decisions in ICN and proposed LAC, a latency-aware cache management policy that bases cache insertion decisions on measurements of residual latency over time on a per-object basis. 
While keeping the same low complexity as standard LRU with probabilistic cache insertion, it provides a finer-granular differentiation of content in terms of expected residual latency. 
Two main advantages have been demonstrated: (i) superior performance in terms of realized delivery time at the end-user plus maximum and average link load reduction, when compared to classical LRU and probabilistic caching approaches; (ii) faster convergence w.r.t. probabilistic caching approaches along with reduced standard deviation. 

We leave for future work a thorough characterization of LAC dynamics, especially in a network of caches, where the coupling with hop-by-hop forwarding may be addressed through a joint optimization. 

The sensitivity to variations in network conditions and routing will also be investigated to highlight the benefit in terms of self-adaptiveness of a measurement-based approach w.r.t. classical latency-insensitive approaches.

\bibliographystyle{abbrv}
\bibliography{references}

\begin{thebibliography}{10}

\bibitem{ICN2014-Kurose}
M.~Badov, A.~Seetharam, J.~Kurose, V.~Firoiu, and S.~Nanda.
\newblock Congestion-aware caching and search in information-centric networks.
\newblock In {\em Proceedings of the 1st International Conference on
  Information-centric Networking}, ICN '14, pages 37--46, New York, NY, USA,
  2014. ACM.

\bibitem{Bianchi:2013:CBS:2500098.2500106}
G.~Bianchi, A.~Detti, A.~Caponi, and N.~Blefari~Melazzi.
\newblock Check before storing: What is the performance price of content
  integrity verification in lru caching?
\newblock {\em SIGCOMM Comput. Commun. Rev.}, 43(3):59--67, July 2013.

\bibitem{DBLP:journals/icl/Blefari-MelazziBCD14}
N.~Blefari{-}Melazzi, G.~Bianchi, A.~Caponi, and A.~Detti.
\newblock A general, tractable and accurate model for a cascade of {LRU}
  caches.
\newblock {\em {IEEE} Communications Letters}, 18(5):877--880, 2014.

\bibitem{breslau:web}
L.~Breslau, P.~Cao, L.~Fan, G.~Phillips, and S.~Shenker.
\newblock Web caching and zipf-like distributions: Evidence and implications.
\newblock In {\em INFOCOM'99}, pages 126--134, 1999.

\bibitem{Carofiglio:2013:PBS:2542828.2542992}
G.~Carofiglio, M.~Gallo, and L.~Muscariello.
\newblock On the performance of bandwidth and storage sharing in
  information-centric networks.
\newblock {\em Comput. Netw.}, 57(17):3743--3758, Dec. 2013.

\bibitem{ITC}
G.~Carofiglio, M.~Gallo, L.~Muscariello, and D.~Perino.
\newblock Modeling data transfer in content-centric networking.
\newblock In {\em Proceedings of the 23rd International Teletraffic Congress},
  ITC '11, pages 111--118. International Teletraffic Congress, 2011.

\bibitem{Che:2006:HWC:2312147.2313846}
H.~Che, Y.~Tung, and Z.~Wang.
\newblock Hierarchical web caching systems: Modeling, design and experimental
  results.
\newblock {\em IEEE J.Sel. A. Commun.}, 20(7):1305--1314, Sept. 2006.

\bibitem{Fricker:2012:VAA:2414276.2414286}
C.~Fricker, P.~Robert, and J.~Roberts.
\newblock A versatile and accurate approximation for lru cache performance.
\newblock In {\em Proceedings of the 24th International Teletraffic Congress},
  ITC '12, pages 8:1--8:8. International Teletraffic Congress, 2012.

\bibitem{DBLP:journals/corr/abs-1202-4880}
M.~Gallo, B.~Kauffmann, L.~Muscariello, A.~Simonian, and C.~Tanguy.
\newblock Performance evaluation of the random replacement policy for networks
  of caches.
\newblock {\em CoRR}, abs/1202.4880, 2012.

\bibitem{LCN2014}
A.~Ioannou and S.~Weber.
\newblock Towards on-path caching alternatives in information-centric networks.
\newblock In {\em Local Computer Networks, 2014 IEEE International Conference
  on}, 2014.

\bibitem{Jacobson:2009:NNC:1658939.1658941}
V.~Jacobson, D.~K. Smetters, J.~D. Thornton, M.~F. Plass, N.~H. Briggs, and
  R.~L. Braynard.
\newblock Networking named content.
\newblock In {\em Proceedings of the 5th International Conference on Emerging
  Networking Experiments and Technologies}, CoNEXT '09, pages 1--12, New York,
  NY, USA, 2009. ACM.

\bibitem{Jelenkovic:2004:OLC:1024662.1024670}
P.~R. Jelenkovi\'{c} and A.~Radovanovi\'{c}.
\newblock Optimizing lru caching for variable document sizes.
\newblock {\em Comb. Probab. Comput.}, 13(4-5):627--643, July 2004.

\bibitem{journals/corr/JinYKSYHL13}
B.~Jin, S.~Yun, D.~Kim, J.~Shin, Y.~Yi, S.~Hong, and B.-J. Lee.
\newblock On the delay scaling laws of cache networks.
\newblock {\em CoRR}, abs/1310.0572, 2013.

\bibitem{1395054}
N.~Laoutaris, S.~Syntila, and I.~Stavrakakis.
\newblock Meta algorithms for hierarchical web caches.
\newblock In {\em Performance, Computing, and Communications, 2004 IEEE
  International Conference on}, pages 445--452, 2004.

\bibitem{DBLP:journals/corr/MartinaGL13}
V.~Martina, M.~Garetto, and E.~Leonardi.
\newblock A unified approach to the performance analysis of caching systems.
\newblock {\em CoRR}, abs/1307.6702, 2013.

\bibitem{Time-Shifted-Simon}
Z.~Ming, X.~Mingwei, and D.~Wang.
\newblock Time-shifted tv in content centric networks: The case for cooperative
  in-network caching.
\newblock In {\em Proceedings of the IEEE International Conference on
  Communications, 2011}, ICC'11, 2011.

\bibitem{Age-Based}
Z.~Ming, X.~Mingwei, and D.~Wang.
\newblock Age-based cooperative caching in information-centric networks.
\newblock In {\em Proceedings of the IEEE Nomen 2012}, Nomen'12, 2012.

\bibitem{Mitra:2011:CWV:1961659.1961662}
S.~Mitra, M.~Agrawal, A.~Yadav, N.~Carlsson, D.~Eager, and A.~Mahanti.
\newblock Characterizing web-based video sharing workloads.
\newblock {\em ACM Trans. Web}, 5(2):8:1--8:27, May 2011.

\bibitem{Podlipnig:2003:SWC:954339.954341}
S.~Podlipnig and L.~B\"{o}sz\"{o}rmenyi.
\newblock A survey of web cache replacement strategies.
\newblock {\em ACM Comput. Surv.}, 35(4):374--398, Dec. 2003.

\bibitem{Prob-Cache}
I.~Psaras, W.~K. Chai, and G.~Pavlou.
\newblock Probabilistic in-network caching for information-centric networks.
\newblock In {\em Proceedings of the Second Edition of the ICN Workshop on
  Information-centric Networking}, ICN '12, pages 55--60, New York, NY, USA,
  2012. ACM.

\bibitem{Rosensweig:2010:AMG:1833515.1833684}
E.~J. Rosensweig, J.~Kurose, and D.~Towsley.
\newblock Approximate models for general cache networks.
\newblock In {\em Proceedings of the 29th Conference on Information
  Communications}, INFOCOM'10, pages 1100--1108, Piscataway, NJ, USA, 2010.
  IEEE Press.

\bibitem{Starobinski:2001:PMW:570289.570293}
D.~Starobinski and D.~Tse.
\newblock Probabilistic methods for web caching.
\newblock {\em Perform. Eval.}, 46(2-3):125--137, Oct. 2001.

\bibitem{Tong:2013:LSS:2505515.2507857}
J.~Tong, G.~Wang, and X.~Liu.
\newblock Latency-aware strategy for static list caching in flash-based web
  search engines.
\newblock In {\em Proceedings of the 22Nd ACM International Conference on
  Information and Knowledge Management}, CIKM '13, pages 1209--1212, New York,
  NY, USA, 2013. ACM.

\bibitem{643}
Y.~Wang, Z.~Li, G.~Tyson, S.~Uhlig, and G.~Xie.
\newblock Optimal cache allocation for content-centric networking.
\newblock In {\em IEEE Intl. Conference on Network Protocols}, 2013.

\bibitem{Congestion-Aware-Gerla}
Y.-T.~Y. Yu, F.~Bronzino, R.~Fan, C.~Westphal, and M.~Gerla.
\newblock Congestion-aware edge caching for adaptive video streaming in
  information-centric networking.

\end{thebibliography}

\end{document}